\documentclass[11pt,a4paper]{article}

\usepackage{fontspec}
\usepackage[titletoc,title,header]{appendix}

\usepackage{amsmath}
\usepackage{amssymb}
\usepackage{amsthm}
\usepackage{mathtools}
\usepackage{xparse}    

\usepackage[margin=1in]{geometry}

\usepackage{graphicx}
\usepackage{xcolor}

\usepackage{hyperref}
\hypersetup{
  colorlinks=true,
  linkcolor=blue,
  citecolor=blue,
  urlcolor=blue,
  pdfauthor={Yoshitsugu Sekine},
  pdftitle={Gaussian \(\beta\)-Markov Path Spaces without a Spectral Gap and the Araki--Woods Representations}
}

\usepackage[numbers]{natbib}
\theoremstyle{plain}
\newtheorem{thm}{Theorem}[section]
\AtEndEnvironment{thm}{\qed}
\newtheorem{lem}[thm]{Lemma}
\AtEndEnvironment{lem}{\qed}
\newtheorem{prop}[thm]{Proposition}
\AtEndEnvironment{prop}{\qed}
\newtheorem{cor}[thm]{Corollary}
\AtEndEnvironment{cor}{\qed}

\AtEndEnvironment{conj}{\qed}

\AtEndEnvironment{fact}{\qed}

\theoremstyle{definition}
\newtheorem{defn}[thm]{Definition}
\AtEndEnvironment{defn}{\qed}
\newtheorem{ex}[thm]{Example}
\AtEndEnvironment{ex}{\qed}

\theoremstyle{remark}
\newtheorem{rem}[thm]{Remark}
\AtEndEnvironment{rem}{\qed}

\makeatletter
\providecommand*{\dashv}{\mathrel{\mathpalette\@Dashv\vDash}}
\newcommand*{\@dashv}[2]{\reflectbox{$\m@th#1#2$}}
\makeatother
\newcommand{\abs}[1]{\left| #1 \right|} 
\NewDocumentCommand{\weaknorm}{O{\dbk} m}{#1{#2}} 
\newcommand{\norm}[1]{\left\Vert #1 \right\Vert} 

\newcommand{\bkt}[2]{\left\langle #1,\,#2 \right\rangle} 
\newcommand{\rbkt}[2]{\left( #1,\,#2 \right)} 
\newcommand{\isomto}{\mathrel{\rightarrowtail\kern-1.9ex\twoheadrightarrow}} 
\newcommand{\cbk}[1]{\left\{ #1 \right\}} 
\newcommand{\dbk}[1]{\left\langle #1 \right\rangle} 
\newcommand{\pairbk}[1]{\rbk{#1}} 
\newcommand{\rbk}[1]{\left( #1 \right)} 
\newcommand{\sqbk}[1]{\left[ #1 \right]} 
\newcommand{\fun}[2]{#1 \rbk{#2}} 
\newcommand{\sqfun}[2]{#1 \sqbk{#2}} 
\newcommand{\closedinterval}[2]{\sqbk{#1,\,#2}} 
\newcommand{\leftopeninterval}[2]{\left(#1, \, #2 \right]} 
\newcommand{\openinterval}[2]{\rbk{#1,\,#2}} 

\NewDocumentCommand{\imunit}{O{\mathsf{i}}}{#1} 
\NewDocumentCommand{\placeholder}{O{\bullet}}{#1} 
\NewDocumentCommand{\trace}{O{\operatorname{Tr}}}{#1} 
\newcommand{\cmpconj}[1]{\overline{#1}} 
\newcommand{\dom}{\operatorname{dom}} 
\newcommand{\eqcsq}[1]{\sqbk{#1}} 
\newcommand{\fml}[2]{\cbk{#1}_{#2}} 
\newcommand{\inv}[1]{#1^{-1}} 
\newcommand{\kroneckerdelta}{\delta} 
\newcommand{\laplacian}{\Delta} 
\newcommand{\napiernum}{\mathsf{e}} 
\newcommand{\opchern}[1]{\operatorname{ch}} 
\newcommand{\opimag}{\operatorname{Im}} 
\newcommand{\opreal}{\operatorname{Re}} 
\newcommand{\pd}[2]{\frac{\partial #1}{\partial #2}} 
\newcommand{\setSymbolDownLeft}[2]{{\vphantom{#2}}_{#1}{#2}} 
\newcommand{\setSymbolUpLeft}[2]{{\vphantom{#2}}^{#1}{#2}} 
\NewDocumentCommand{\agvariety}{O{\mathcal}}{#1} 
\NewDocumentCommand{\cmdrel}{O{\omega}}{#1} 
\NewDocumentCommand{\dfsp}{O{A}}{#1} 
\NewDocumentCommand{\eqcpointed}{O{\eqcsq} m}{#1{#2}_{\ast}} 
\NewDocumentCommand{\fnheaviside}{O{H}}{#1} 
\NewDocumentCommand{\fthol}{O{\mathcal{O}}}{#1} 
\NewDocumentCommand{\ftmero}{O{\mathcal{M}}}{#1} 
\NewDocumentCommand{\grcentralizer}{O{Z}}{#1} 
\NewDocumentCommand{\grmetform}{O{2} m m}{\grmet[#1] \! \rbkt{#2}{#3}} 
\NewDocumentCommand{\grmet}{O{2}}{\setSymbolDownLeft{#1}{g}} 
\NewDocumentCommand{\grnormalizer}{O{N}}{#1} 
\NewDocumentCommand{\gropasym}{O{A}}{#1} 
\NewDocumentCommand{\gropsym}{O{S}}{#1} 
\NewDocumentCommand{\grpermorderedpair}{O{\mathcal{P}}}{#1} 
\NewDocumentCommand{\grsym}{O{\mathfrak{S}} m}{#1_{#2}} 
\NewDocumentCommand{\gtbase}{O{\mathcal}}{#1} 
\NewDocumentCommand{\gtfilter}{O{\mathcal}}{#1} 
\NewDocumentCommand{\gtfmlclosed}{O{\mathcal}}{#1} 
\NewDocumentCommand{\gtfmlopen}{O{\mathcal}}{#1} 
\NewDocumentCommand{\gtopenball}{O{U}}{#1} 
\NewDocumentCommand{\gtopencover}{O{\mathcal}}{#1} 
\NewDocumentCommand{\gtopennbh}{O{\mathcal}}{#1} 
\NewDocumentCommand{\gtpreopencover}{O{\mathcal}}{#1} 
\NewDocumentCommand{\gtsubbase}{O{\mathcal}}{#1} 
\NewDocumentCommand{\gtvicinity}{O{\mathcal}}{#1} 
\NewDocumentCommand{\lasp}{O{\mathcal}}{#1} 
\NewDocumentCommand{\latprightrbk}{O{\top} m}{\rbk{#2}^{#1}} 
\NewDocumentCommand{\latpright}{O{\top} m}{#2^{#1}} 
\NewDocumentCommand{\latp}{O{t} m}{\setSymbolUpLeft{#1}{#2}} 
\NewDocumentCommand{\lpdistribution}{O{\mu} m}{#2_{\ast,#1}} 
\NewDocumentCommand{\lpmollifier}{O{\rho}}{#1} 
\NewDocumentCommand{\lpofpositive}{O{\chi}}{#1} 
\NewDocumentCommand{\manliederiv}{O{L}}{#1} 
\NewDocumentCommand{\mansmoothnbh}{O{\mathcal}}{#1} 
\NewDocumentCommand{\mblfmldsysgenerated}{O{d} m}{\fun{#1}{#2}} 
\NewDocumentCommand{\mblfmlgenerated}{O{\sigma} m}{\fun{#1}{#2}} 
\NewDocumentCommand{\oacorrfn}{O{\Gamma}}{#1} 
\NewDocumentCommand{\oagnsvector}{O{\Omega}}{#1} 
\NewDocumentCommand{\oaideal}{O{\mathcal}}{#1} 
\NewDocumentCommand{\oanumberoperator}{O{A}}{#1} 
\NewDocumentCommand{\oaposcone}{O{\mathcal{P}}}{#1} 
\NewDocumentCommand{\oapressure}{O{P}}{#1} 
\NewDocumentCommand{\oarepn}{O{\pi}}{#1} 
\NewDocumentCommand{\oaspkmsstate}{O{\txtkms}}{#1} 
\NewDocumentCommand{\oaspnormalstate}{O{N}}{#1} 
\NewDocumentCommand{\oasppurestate}{O{P}}{#1} 
\NewDocumentCommand{\oaspstate}{O{E}}{#1} 
\NewDocumentCommand{\oastatevector}{O{\Omega}}{#1} 
\NewDocumentCommand{\oastate}{O{\omega}}{#1} 
\NewDocumentCommand{\opdilation}{O{\delta}}{#1} 
\NewDocumentCommand{\opdmat}{O{\rho}}{#1} 
\NewDocumentCommand{\opfockan}{O{a}}{#1} 
\NewDocumentCommand{\opfockcran}{O{a}}{#1^{\#}} 
\NewDocumentCommand{\opfockcrdagger}{O{a}}{#1^{\dagger}} 
\NewDocumentCommand{\opfockcr}{O{a}}{#1^{\ast}} 
\NewDocumentCommand{\opfocknumber}{O{N}}{#1} 
\NewDocumentCommand{\opfocksegalconj}{O{\pi}}{#1} 
\NewDocumentCommand{\opfocksegal}{O{\phi}}{#1} 
\NewDocumentCommand{\opfockvac}{O{\Omega}}{#1} 
\NewDocumentCommand{\opspecmeas}{O{E}}{#1} 
\NewDocumentCommand{\opspec}{O{} m}{\fun{\sigma_{#1}}{#2}} 
\NewDocumentCommand{\optransl}{O{\tau}}{#1} 
\NewDocumentCommand{\opvarspec}{O{} m}{\operatorname{spec}_{#1} #2} 
\NewDocumentCommand{\physaction}{O{\mathcal{A}}}{#1} 
\NewDocumentCommand{\physcharge}{O{e}}{\mathrm{#1}} 
\NewDocumentCommand{\physcplconst}{O{\mathsf{g}}}{#1} 
\NewDocumentCommand{\physelectrostaticcapasity}{O{\mathrm{Cap}}}{#1} 
\NewDocumentCommand{\physenergy}{O{E}}{#1} 
\NewDocumentCommand{\physgse}{O{E}}{#1_{0}} 
\NewDocumentCommand{\physham}{O{H}}{#1} 
\NewDocumentCommand{\physlagdensity}{O{\mathcal{L}}}{#1} 
\NewDocumentCommand{\physlag}{O{L}}{#1} 
\NewDocumentCommand{\physliouvilean}{O{L}}{#1} 
\NewDocumentCommand{\physmass}{O{m}}{#1} 
\NewDocumentCommand{\prbbrownmv}{O{B}}{#1} 
\NewDocumentCommand{\prbcharfun}{O{\chi}}{#1} 
\NewDocumentCommand{\prbdist}{O{\mathcal{P}}}{#1} 
\NewDocumentCommand{\prbgaussianmeasure}{O{\msrcal{N}}}{#1} 
\NewDocumentCommand{\prbnormaldist}{O{N}}{#1} 
\NewDocumentCommand{\prbpoissonprocess}{O{N}}{#1} 
\NewDocumentCommand{\prbprocess}{O{X}}{#1} 
\NewDocumentCommand{\prbqspace}{O{\mathcal{Q}}}{#1} 
\NewDocumentCommand{\prbspsample}{O{\Omega}}{#1} 
\NewDocumentCommand{\prbspstate}{O{\mathsf{S}}}{#1} 
\NewDocumentCommand{\psh}{O{\mathfrak}}{#1} 
\NewDocumentCommand{\qtquantumchannel}{O{\mathcal{L}}}{#1} 
\NewDocumentCommand{\repn}{O{\pi}}{#1} 
\NewDocumentCommand{\schattencls}{O{\mathbb{K}}}{#1} 
\NewDocumentCommand{\setfmlcylinder}{O{\mathcal{C}}}{#1} 
\NewDocumentCommand{\setfml}{O{\mathcal}}{#1} 
\NewDocumentCommand{\setindex}{O{\mathcal} m}{#1{#2}} 
\NewDocumentCommand{\setlattice}{O{\Gamma}}{#1} 
\NewDocumentCommand{\setspecial}{O{\mathcal} m}{#1{#2}} 
\NewDocumentCommand{\shdiffform}{O{\sheaf{A}}}{#1} 
\NewDocumentCommand{\sheaf}{O{\mathfrak}}{#1} 
\NewDocumentCommand{\smchemicalpotential}{O{\mu}}{#1} 
\NewDocumentCommand{\smenergydensity}{O{\varrho}}{#1} 
\NewDocumentCommand{\smfluctuationwithdmat}{O{\beta} m}{\smuncertaintywithdmat[#1]{#2}^2} 
\NewDocumentCommand{\sminvtemperature}{O{\beta}}{#1} 
\NewDocumentCommand{\smlocaldensityoperator}{O{\rho}}{#1} 
\NewDocumentCommand{\smmicrocanonicalstate}{O{\beta} m}{\physmean{#2}_{#1}} 
\NewDocumentCommand{\smnumberdensity}{O{\rho}}{#1} 
\NewDocumentCommand{\smooth}{O{\mathcal{E}}}{#1} 
\NewDocumentCommand{\smparticlenumber}{O{N}}{#1} 
\NewDocumentCommand{\smpressure}{O{p}}{#1} 
\NewDocumentCommand{\smspecificfreeenergy}{O{\bar{f}}}{#1} 
\NewDocumentCommand{\smthermalvac}{O{\beta}}{\Omega_{#1}} 
\NewDocumentCommand{\smuncertaintywithdmat}{O{\beta} m}{\rbk{\triangle #2}_{#1}} 
\NewDocumentCommand{\sphilbfrak}{O{\mathfrak}}{#1} 
\NewDocumentCommand{\sphilb}{O{\mathcal}}{#1} 
\NewDocumentCommand{\splowerhalf}{O{\mathbb{H}}}{#1_{\txtneg}} 
\NewDocumentCommand{\spupperhalf}{O{\mathbb{H}}}{#1_{\txtnonneg}} 
\NewDocumentCommand{\topmetric}{O{d}}{#1} 
\NewDocumentCommand{\vaoutnormal}{O{\widehat}}{#1} 
\newcommand{\category}[1]{\mathop{\mathsf{#1}}} 

\newcommand{\catpresheaf}[1]{\category{PSh}} 
\newcommand{\faadj}[1]{#1^{\ast}} 
\newcommand{\fafouriertr}{\mathcal{F}} 
\newcommand{\fldcmp}{\fld{C}} 
\newcommand{\fldreal}{\fld{R}} 
\newcommand{\fld}[1]{\mathbb{#1}} 
\newcommand{\fndef}[1]{\boldsymbol{1}_{#1}} 
\newcommand{\fnexp}[1]{\fun{\exp}{#1}} 
\newcommand{\lpseq}{\ell} 
\newcommand{\lp}{L} 
\newcommand{\mblfmlfrak}[1]{\mathfrak{#1}} 
\newcommand{\monnat}{\mathbb{N}} 
\newcommand{\msrcal}[1]{\mathcal{#1}} 
\newcommand{\msrdirac}{\delta} 
\newcommand{\msrprb}{\mathrm{Pr}} 

\newcommand{\oaresolventalgebra}{\oa{R}} 
\newcommand{\oaresolvent}{R} 
\newcommand{\oa}[1]{\mathcal{#1}} 
\newcommand{\opdmsr}[1]{\mathop{d #1}} 
\newcommand{\opfocksndqntdiff}{d \Gamma} 
\newcommand{\opfocksndqnt}{\Gamma} 
\newcommand{\opfockweyl}{W} 
\newcommand{\opform}[1]{\mathsf{#1}} 
\newcommand{\opspecint}[1]{\mathcal{E}} 
\newcommand{\physmean}[1]{\dbk{#1}} 
\newcommand{\prbexp}{\mathbb{E}} 
\newcommand{\pushoutrbk}[1]{\rbk{#1}_{\ast}} 
\newcommand{\ringratint}{\mathbb{Z}} 
\newcommand{\semigrposint}{\monnat_1} 

\newcommand{\seq}[2]{\if\relax\detokenize{#1}\relax \rbk{#1} \else \rbk{#1}_{#2} \fi} 
\newcommand{\setisomorphism}[1]{\operatorname{Iso}} 
\newcommand{\setone}[1]{\cbk{#1}}
\newcommand{\set}[2]{\left\{#1 \, \middle| \, #2\right\}}
\newcommand{\spfockgamma}{\Gamma} 
\newcommand{\txtarakiwoods}{\mathrm{AW},\mathrm{b}} 
\newcommand{\txtbsn}{\mathrm{b}} 
\newcommand{\txteuclid}{\mathrm{E}} 
\newcommand{\txthilbertschmidt}{\mathrm{HS}} 
\newcommand{\txtkms}{\mathrm{KMS}} 
\newcommand{\txtleft}{\mathrm{l}} 
\newcommand{\txtneg}{\mathrm{-}} 
\newcommand{\txtnonneg}{\mathrm{+}} 
\newcommand{\txtright}{\mathrm{r}} 

\title{Gaussian \(\beta\)-Markov Path Spaces without a Spectral Gap and the Araki--Woods Representations}

\author{%
Yoshitsugu Sekine\\{\small\texttt{4429sekine@gmail.com}}%
}

\date{\today}

\begin{document}

\maketitle

\begin{abstract}
Dereziński and Gérard constructed the Gaussian \(\beta\)-Markov path space of a free Bose field in a form equivalent to the Araki--Woods representation, which suits nonequilibrium statistical mechanics with a quantum-field environment, but only for one-particle Hamiltonians with a spectral gap. We extend this construction, together with the equivalence to the Araki--Woods representation, to one-particle Hamiltonians without a spectral gap. For one-particle Hamiltonians joined by isometric embeddings that intertwine their spectral measures and test spaces, as for nested momentum lattices, and for compatible regularizers, the path spaces form a projective system whose limit carries the time translations and the reflection, which gives a quasi-local construction. The extension is intended for constructive quantum field theory at positive temperature and equilibrium statistical mechanics, for instance in the problem of the absence of Bose--Einstein condensation of quasiparticles whose number is not conserved.

\noindent\textbf{Keywords:} \(\beta\)-Markov path space, Gaussian measure, KMS state, Araki--Woods representation, resolvent algebra, infrared singularity, Bose--Einstein condensation
\end{abstract}

\setcounter{tocdepth}{2}
\tableofcontents

\section{Introduction}\label{introduction}

Klein and Landau \cite{KleinLandau001} established the correspondence between stochastically positive KMS systems and periodic stochastic processes \cite[Theorem 6.1, p.~385]{KleinLandau001}, related the two-sided Markov property for semicircles to a cyclicity property of the KMS state \cite[Theorems 11.2 and 11.4]{KleinLandau001}, and identified the modular conjugation with a reflection of the process \cite[Section 12]{KleinLandau001}. These results hold for general stochastically positive KMS systems, and their generality makes them important. Dereziński and Jakšić proved return to equilibrium for a large class of Pauli--Fierz systems at small coupling \cite{DerezinskiJaksic002}, a problem of nonequilibrium quantum statistical mechanics in which the environment of a small quantum system is a free Bose field and the interaction is a perturbation linear in the field \cite[Section 3.1]{DerezinskiJaksic002}. The textbook \cite{DerezinskiGerard001} of Dereziński and Gérard develops the framework of Klein--Landau for quantum fields at positive temperature in a form suited to such problems of nonequilibrium statistical mechanics whose environment is a quantum field. In particular, it constructs the Gaussian \(\sminvtemperature\)-Markov path space of a free Bose field in a form equivalent to the Araki--Woods representation, which is the representation on which these systems are analyzed: a unitary map from its physical Hilbert space onto the Araki--Woods Fock space intertwines the time-zero fields with the Araki--Woods Weyl operators \cite[Theorem 21.58, p.~635]{DerezinskiGerard001}. This equivalence connects the stochastic process with the operator-algebraic description of the environment, and it is this construction that the present paper extends. The Araki--Woods representations are treated in \cite[Section 17.1.5, pp.~437--440]{DerezinskiGerard001}, and their use in \cite{DerezinskiJaksic002} is recorded in \cite[Section 17.4, p.~473]{DerezinskiGerard001}. In the formulation of \cite[Section 21.4]{DerezinskiGerard001}, for an inverse temperature \(\sminvtemperature>0\), a \(\sminvtemperature\)-Markov path space is a probability space with a time-zero \(\sigma\)-algebra, a \(\sminvtemperature\)-periodic group of measure-preserving automorphisms, and a reflection, which satisfy the \(\sminvtemperature\)-reflection property and the Markov property for the two half circles; the reconstruction theorem \cite[Theorems 21.43 and 21.47, pp.~625--626]{DerezinskiGerard001} gives a physical Hilbert space, a Liouvillean, a modular conjugation, and a \(\sminvtemperature\)-KMS state. For a one-particle Hamiltonian \(\physham[h]\), the Gaussian example has the covariance operator \(\rbk{-\partial_t^2+\physham[h]^2}^{-1}\) with \(\sminvtemperature\)-periodic boundary conditions in the imaginary time \(t\); we use the Segal normalization, in which the field has covariance \(\frac{1}{2}\bkt{\cdot}{\rbk{-\partial_t^2+\physham[h]^2}^{-1}\cdot}\), and its physical Hilbert space carries the Araki--Woods representation \cite{ArakiWoods001} of the free Bose field at the Planck density \cite[Sections 21.4.5 and 21.4.6, Theorem 21.58, pp.~633--636]{DerezinskiGerard001}. The treatment of positive-temperature path spaces there follows Klein--Landau and Gérard--Jäkel \cite[Section 21.6, p.~640]{DerezinskiGerard001}, \cite{GerardJakel001}. Although its stated hypothesis admits them, the construction there does not cover singular one-particle Hamiltonians, that is, nonnegative injective ones with \(\inf\opvarspec{\physham[h]}=0\), where \(\opvarspec{\cdot}\) denotes the spectrum, as explained in the next paragraph. For several problems, including return to equilibrium, this restriction is harmless. In \cite{DerezinskiJaksic002}, the one-particle Hamiltonian of the reservoir, the operator \(\rbk{-\laplacian}^{\frac{1}{2}}\) on \(\fun{\lp^2}{\fldreal^d}\), has no spectral gap, but the analysis does not use a path space: it is carried out on the Araki--Woods representation of the reservoir and the Liouvillean of the coupled system, and the field enters only through a form factor that satisfies an infrared condition \cite[opening paragraph of Section 3, Sections 3.2--3.5, Assumption 3.B, and Theorem 3.4]{DerezinskiJaksic002}. The path-space formulation is also useful in constructive quantum field theory at positive temperature and in equilibrium statistical mechanics, for instance in discussing the problem, well known in theoretical physics, that quasiparticles whose number is not conserved do not undergo Bose--Einstein condensation in equilibrium \cite[Section 4.2]{VIYukalov001}. The quasiparticles there have the one-particle Hamiltonian \(A\rbk{-\laplacian}^{\frac{n}{2}}\), \(A,n>0\), and their density is the Bose integral without a chemical potential \cite[(21) and (32)]{VIYukalov001}, so their one-particle Hamiltonian has no spectral gap. In contexts where infrared singularities or the occurrence or absence of Bose--Einstein condensation are discussed, the restriction is unsatisfactory, since these are governed by the behavior of \(\physham[h]\) near zero, for instance through the zero Matsubara mode \(\physham[h]^{-2}\) of the covariance. The aim of this paper is to remove this restriction, and the extension is formulated for use in these discussions.

This paper treats one-particle Hamiltonians \(\physham[h]\) that are nonnegative and injective, without the assumption \(\inf\opvarspec{\physham[h]}>0\). The basic example is the operator \(\rbk{-\laplacian}^{\frac{s}{2}}\), \(s>0\), on \(\fun{\lp^2}{\fldreal^d}\); for \(s=2\) it is the one-particle Hamiltonian \(-\laplacian\) of the free nonrelativistic Bose gas at zero chemical potential. In this case the zero Matsubara mode \(\physham[h]^{-2}\) of the covariance is unbounded, so the covariance is an unbounded operator on the spacetime \(\lp^2\) space (Lemma \ref{lem:covariance-fourier}). The time-zero fields are then defined only for one-particle vectors in a completion of \(\dom\physham[h]^{-1}\) (Definition \ref{def:sharp-time-vectors}), and the time-zero algebras have to be generated by a test space that is compatible with the boundary-value problem on the half circles. In \cite[Section 21.4.5, p.~633]{DerezinskiGerard001} the one-particle Hamiltonian is explicitly assumed to be a self-adjoint operator \(\physham[h]>0\) (denoted \(\epsilon\) there) on a real Hilbert space, and by \cite[Definition 2.58, p.~47]{DerezinskiGerard001} the notation \(B>0\) means that \(B\) is a self-adjoint operator with \(B\geq0\) and \(0\) is not an eigenvalue of \(B\). Thus the stated hypothesis is that \(\physham[h]\) is nonnegative and injective, and it does not require \(\inf\opvarspec{\physham[h]}>0\). The Gaussian \(\lp^2\) space \cite[(21.54), p.~634]{DerezinskiGerard001} of the construction, however, is the Gaussian \(\lp^2\) space over the spacetime \(\lp^2\) space with covariance \(\rbk{-\partial_t^2+\physham[h]^2}^{-1}\) in the sense of \cite[Definition 5.72, pp.~133--134]{DerezinskiGerard001}, which is defined by the weak characteristic functional \(f\mapsto\fnexp{-\frac{1}{2}\bkt{f}{\rbk{-\partial_t^2+\physham[h]^2}^{-1}f}}\) on the whole spacetime \(\lp^2\) space, used as the pivot space \cite[(5.24), p.~133]{DerezinskiGerard001}. This functional is defined on the whole space only if the square root of the covariance operator is everywhere defined. By \eqref{eq:covariance-square-root}, this means that \(\dom\physham[h]^{-1}\) is the whole one-particle space, so that \(\physham[h]^{-1}\) is bounded by the closed graph theorem, that is, \(\inf\opvarspec{\physham[h]}>0\). For the same reason, the smooth compactly supported \(\dom\physham[h]\)-valued functions of \cite[Proposition 21.54(2), p.~634]{DerezinskiGerard001} need not belong to the index space of the field in \cite[Section 21.4.5, p.~633]{DerezinskiGerard001}, the completion of \(\dom\rbk{-\partial_t^2+\physham[h]^2}^{-\frac{1}{2}}\) in the norm \(\norm{\rbk{-\partial_t^2+\physham[h]^2}^{-\frac{1}{2}}\cdot}\), when \(\inf\opvarspec{\physham[h]}=0\), because their time averages need not lie in \(\dom\physham[h]^{-1}\); for the basic model of Example \ref{ex:basic-model}, the vector with Fourier transform \(k\mapsto\fun{\fndef{\setone{\abs{k}\leq1}}}{k}\abs{k}^{s-\frac{d}{2}}\) lies in \(\dom\physham[h]\) but not in \(\dom\physham[h]^{-1}\). Moreover, the construction in \cite[Sections 21.4.5 and 21.4.6]{DerezinskiGerard001} uses an unspecified concrete realization of the Gaussian \(\lp^2\) space and refers its proofs to the zero-temperature arguments of \cite[Section 21.2]{DerezinskiGerard001}. Following \cite[Section 1.6.2, pp.~95--97]{LorincziHiroshimaBetz003}, we therefore index the field by the completion of \(\dom\rbk{-\partial_t^2+\physham[h]^2}^{-\frac{1}{2}}\) used in \cite[Section 21.4.5]{DerezinskiGerard001}, realized as the completion of the domain of the covariance operator in the covariance inner product (Definition \ref{def:covariance-hilbert-space}), and realize it on the negative space of a Hilbert scale with this completion as the pivot space, so that neither the spacetime \(\lp^2\) space nor the unbounded covariance operator enters the definition of the Gaussian measure. The present paper gives complete arguments when the spectrum of \(\physham[h]\) accumulates at zero.

Theorem \ref{thm:singular-gaussian-beta-markov} constructs the path space as a Borel Gaussian measure on a separable Hilbert space defined by a regularizer that commutes with the circle translations and the reflection, so that the path-space automorphisms are compositions with orthogonal maps. Its Markov property follows from the explicit half-circle projection formula of Proposition \ref{prop:half-circle-projections}. This formula involves only the bounded operators \(\fun{\sinh}{\rbk{\frac{\sminvtemperature}{2}+\sigma}\physham[h]}/\fun{\sinh}{\frac{\sminvtemperature\physham[h]}{2}}\) and \(\fun{\sinh}{-\sigma\physham[h]}/\fun{\sinh}{\frac{\sminvtemperature\physham[h]}{2}}\), \(-\frac{\sminvtemperature}{2}\leq\sigma\leq0\), and it requires the test space to be invariant under them (Definition \ref{def:test-space}).

Theorem \ref{thm:araki-woods-equivalence} gives the unitary equivalence of the physical Hilbert space with the Araki--Woods Fock space. The unitary intertwines the time-zero and time-\(\frac{\sminvtemperature}{2}\) Weyl operators with the left and right Araki--Woods Weyl operators, the Euclidean resolvents with the left and right representations of the resolvent algebra \cite{BuchholzGrundling002}, and the Liouvilleans and the modular conjugations. Proposition \ref{prop:imaginary-time-correlations} gives the imaginary-time correlation identities, and Corollary \ref{cor:planck-domain-extension} extends both results to the domain of the square root of the Planck density. Since Theorem \ref{thm:singular-gaussian-beta-markov} verifies the axioms of a \(\sminvtemperature\)-Markov path space, the perturbation theory of \cite[Section 21.5, pp.~636--639]{DerezinskiGerard001} applies without change. A real potential \(V\) measurable with respect to the time-zero \(\sigma\)-algebra such that \(V\) and \(\napiernum^{-\sminvtemperature V}\) are integrable gives a perturbed \(\sminvtemperature\)-Markov path space \cite[Proposition 21.60, p.~637]{DerezinskiGerard001}, and under the additional integrability conditions there, the perturbed Liouvilleans are expressed through the free Liouvillean \cite[Theorems 21.64--21.66, p.~639]{DerezinskiGerard001}, which Theorem \ref{thm:araki-woods-equivalence} carries to the Araki--Woods Liouvillean of Definition \ref{def:araki-woods}.

Theorem \ref{thm:quasi-local} gives the quasi-local construction for a sequence of one-particle Hamiltonians joined by isometric embeddings that intertwine their spectral measures and preserve the test spaces, such as the nested momentum lattices of the finite-volume free Bose gas (Example \ref{ex:momentum-lattices}). With compatible regularizers, the adjoints of the induced isometries of the covariance Hilbert spaces extend to continuous maps between the carriers, under which the Gaussian measures, the sharp-time fields, the \(\sigma\)-algebras, the time translations, and the reflection are consistent, and the path spaces have a projective-limit measure, on which the local observables generate a quasi-local algebra.

Section \ref{sec:setting} fixes the one-particle Hamiltonian and the covariance. Section \ref{sec:sharp-time} constructs the sharp-time vectors and proves the half-circle projection formula. Section \ref{sec:markov} constructs the Gaussian measure and proves the \(\sminvtemperature\)-Markov property. Section \ref{sec:araki-woods} proves the Euclidean realization of the Araki--Woods representations. Section \ref{sec:quasi-local} gives the quasi-local construction.

\section{One-particle Hamiltonian and covariance}\label{sec:setting}

In this section, we fix the one-particle Hamiltonian and the imaginary-time circle, and we describe the covariance operator in the Matsubara representation.

Throughout, \(\monnat\) is the set of nonnegative integers, \(\semigrposint=\monnat\setminus\setone{0}\), and \(\ringratint\) is the set of integers. For a set \(E\), \(\fndef{E}\) is its indicator function; for a Hilbert space \(\sphilbfrak{k}\), \(\fun{\lpseq^2}{\ringratint;\sphilbfrak{k}}\) is the Hilbert space of square-summable \(\sphilbfrak{k}\)-valued sequences indexed by \(\ringratint\), and \(\norm{\cdot}_{\txthilbertschmidt}\) is the Hilbert--Schmidt norm. For a probability measure \(\msrprb\), \(\prbexp_{\msrprb}\) is the expectation with respect to \(\msrprb\), and for a measurable map \(\tau\), \(\pushoutrbk{\tau}\msrprb=\msrprb\circ\tau^{-1}\) is the image measure. The Fourier transform on \(\fldreal^d\) carries \(2\pi\) in the exponent: \(\fun{\rbk{\fafouriertr f}}{k}=\int_{\fldreal^d}\napiernum^{-2\pi\imunit k\cdot x}\fun{f}{x}\opdmsr{x}\) for \(f\in\fun{\lp^1}{\fldreal^d}\cap\fun{\lp^2}{\fldreal^d}\). It extends to a unitary operator \(\fafouriertr\) on \(\fun{\lp^2}{\fldreal^d}\), and \(\fun{\rbk{\fafouriertr\partial_jf}}{k}=2\pi\imunit k_j\fun{\rbk{\fafouriertr f}}{k}\) for Schwartz functions \(f\) \cite[Theorems 8.22(e), p.~249, and 8.29, p.~252]{GeraldFolland001}, so that \(\fafouriertr\) carries \(-\laplacian\) to the multiplication by \(4\pi^2\abs{k}^2\). The Fourier coefficients in the imaginary time, defined in \eqref{eq:fourier-coefficient}, also carry \(2\pi\) in the exponent through the Matsubara frequencies \(\nu_n=2\pi n/\sminvtemperature\) introduced below.

\subsection{The one-particle Hamiltonian}\label{the-one-particle-hamiltonian}

In this subsection, we state the standing assumption on the one-particle Hamiltonian and the basic example.

\begin{defn}[One-particle Hamiltonian]\label{def:one-particle-hamiltonian}
Let $\sphilbfrak{h}\neq\setone{0}$ be a separable complex Hilbert space, and let $\physham[h]$ be an injective self-adjoint
operator on $\sphilbfrak{h}$ with $\bkt{u}{\physham[h]u}_{\sphilbfrak{h}}\geq0$ for $u\in\dom\physham[h]$.
Let $\sphilbfrak{h}_{\fldreal}$ be the real Hilbert space with the underlying real vector space of $\sphilbfrak{h}$ and the inner
product $\bkt{f}{g}_{\sphilbfrak{h}_{\fldreal}}=\opreal\bkt{f}{g}_{\sphilbfrak{h}}$.
For a real Borel function $\varphi$ on $\openinterval{0}{\infty}$, $\fun{\varphi}{\physham[h]}$ denotes the operator of the
functional calculus on $\sphilbfrak{h}$ and also the same map regarded as a real linear operator on $\sphilbfrak{h}_{\fldreal}$.
\end{defn}

For a self-adjoint operator \(A\) on a real or complex Hilbert space, \(\msrprb_A\) denotes its spectral measure, the projection-valued measure on the Borel sets of \(\fldreal\) with \(\fun{\varphi}{A}=\int\varphi\,\opdmsr{\msrprb_A}\) for real Borel functions \(\varphi\), and \(\fun{\msrprb_A}{E}\) is the spectral projection of \(A\) for a Borel set \(E\); we use this notation without further comment. The spectral measure \(\msrprb_{\physham[h]}\) is carried by \(\openinterval{0}{\infty}\), because \(\physham[h]\) is injective and nonnegative, so \(\fun{\varphi}{\physham[h]}\) is well defined for Borel functions on \(\openinterval{0}{\infty}\). The real linear operator \(\fun{\varphi}{\physham[h]}\) on \(\sphilbfrak{h}_{\fldreal}\) is self-adjoint, and it is the functional calculus of the real operator \(\physham[h]\) on \(\sphilbfrak{h}_{\fldreal}\). The Gaussian field of Section \ref{sec:markov} is indexed by the real space \(\sphilbfrak{h}_{\fldreal}\), so that both real components of a complex field are random variables.

\begin{ex}[Basic model]\label{ex:basic-model}
For $d\in\semigrposint$ and $s>0$, let $\sphilbfrak{h}=\fun{\lp^2}{\fldreal^d;\fldcmp}$ with the Fourier transform $\fafouriertr$ fixed at the
beginning of this section.
The operator $\physham[h]=\rbk{-\laplacian}^{\frac{s}{2}}$, that is, $\fafouriertr\physham[h]f=\abs{2\pi k}^s\fafouriertr f$ on its maximal domain,
satisfies Definition \ref{def:one-particle-hamiltonian}, and $\inf\opvarspec{\physham[h]}=0$.
\end{ex}

\subsection{The time circle and the covariance}\label{the-time-circle-and-the-covariance}

In this subsection, we introduce the circle, the Matsubara frequencies, and the covariance operator, and we compute the covariance and its square root in the Fourier coefficients.

Fix \(\sminvtemperature>0\). Let \(S_{\sminvtemperature}=\fldreal/\sminvtemperature\ringratint\) with the Lebesgue measure, represented by \(\leftopeninterval{-\frac{\sminvtemperature}{2}}{\frac{\sminvtemperature}{2}}\), and let \begin{equation}\label{eq:circle-distance-matsubara}
\begin{aligned}
\fun{d_{\sminvtemperature}}{s,t}
&=
\min\setone{\abs{t-s},\sminvtemperature-\abs{t-s}},
&
\nu_n
&=
\frac{2\pi n}{\sminvtemperature},
\quad n\in\ringratint.
\end{aligned}
\end{equation} The real spacetime Hilbert space is \(\fun{\lp^2}{S_{\sminvtemperature};\sphilbfrak{h}_{\fldreal}}\); as a set it is \(\fun{\lp^2}{S_{\sminvtemperature};\sphilbfrak{h}}\), and its inner product is the real part of that of \(\fun{\lp^2}{S_{\sminvtemperature};\sphilbfrak{h}}\). Throughout the paper, we use without comment the standard unitary identifications \(\fun{\lp^2}{S_{\sminvtemperature};\sphilbfrak{h}}\cong\fun{\lp^2}{S_{\sminvtemperature};\fldcmp}\otimes\sphilbfrak{h}\) and \(\fun{\lp^2}{S_{\sminvtemperature};\sphilbfrak{h}_{\fldreal}}\cong\fun{\lp^2}{S_{\sminvtemperature};\fldreal}\otimes\sphilbfrak{h}_{\fldreal}\), under which \(\varphi\otimes g\) is the function \(t\mapsto\fun{\varphi}{t}g\), and similar well-known identifications of vector-valued \(\lp^2\) spaces, as the context requires. For \(F\in\fun{\lp^2}{S_{\sminvtemperature};\sphilbfrak{h}_{\fldreal}}\) and \(n\in\ringratint\), the Fourier coefficient \begin{equation}\label{eq:fourier-coefficient}
\begin{aligned}
F_n
&=
\frac{1}{\sqrt{\sminvtemperature}}\int_{S_{\sminvtemperature}}\napiernum^{-\imunit\nu_nt}\fun{F}{t}\opdmsr{t}
\in\sphilbfrak{h}
\end{aligned}
\end{equation} is taken with the complex scalar multiplication of \(\sphilbfrak{h}\). The map \(F\mapsto\seq{F_n}{n\in\ringratint}\) is a unitary map from \(\fun{\lp^2}{S_{\sminvtemperature};\sphilbfrak{h}}\) onto \(\fun{\lpseq^2}{\ringratint;\sphilbfrak{h}}\), hence \begin{equation}\label{eq:real-parseval}
\begin{aligned}
\bkt{F}{G}_{\fun{\lp^2}{S_{\sminvtemperature};\sphilbfrak{h}_{\fldreal}}}
&=
\sum_{n\in\ringratint}\opreal\bkt{F_n}{G_n}_{\sphilbfrak{h}},
\quad F,G\in\fun{\lp^2}{S_{\sminvtemperature};\sphilbfrak{h}_{\fldreal}}.
\end{aligned}
\end{equation}

Let \(D_t^2\) be the closure of \(-\pd{^2}{t^2}\) on the \(\sminvtemperature\)-periodic smooth \(\sphilbfrak{h}_{\fldreal}\)-valued functions, let \(\physham[h]^2\) also denote the operator \(\rbk{\physham[h]^2F}(t)=\physham[h]^2\fun{F}{t}\) on its maximal domain in \(\fun{\lp^2}{S_{\sminvtemperature};\sphilbfrak{h}_{\fldreal}}\), and let \begin{equation}\label{eq:covariance-operator}
\begin{aligned}
C
&=
\rbk{D_t^2+\physham[h]^2}^{-1}.
\end{aligned}
\end{equation}

\begin{lem}[Covariance in the Matsubara representation]\label{lem:covariance-fourier}
The operator $D_t^2+\physham[h]^2$ on $\dom D_t^2\cap\dom\physham[h]^2$ is nonnegative, injective, and self-adjoint, and
$\rbk{\rbk{D_t^2+\physham[h]^2}F}_n=\rbk{\nu_n^2+\physham[h]^2}F_n$.
Its inverse $C$ is nonnegative, injective, self-adjoint, and densely defined, and
\begin{equation}\label{eq:covariance-square-root}
\begin{aligned}
\dom C^{\frac{1}{2}}
&=
\set{F\in\fun{\lp^2}{S_{\sminvtemperature};\sphilbfrak{h}_{\fldreal}}}{F_0\in\dom\physham[h]^{-1}},
&
\rbk{C^{\frac{1}{2}}F}_n
&=
\rbk{\nu_n^2+\physham[h]^2}^{-\frac{1}{2}}F_n.
\end{aligned}
\end{equation}
The operator $C$ is bounded if and only if $\inf\opvarspec{\physham[h]}>0$.
\end{lem}

\begin{proof}
The Fourier transform in $t$ diagonalizes $D_t^2$ as multiplication by $\nu_n^2$ and maps $\physham[h]^2$ to
$\bigoplus_n\physham[h]^2$, so $D_t^2+\physham[h]^2$ is the real linear operator of $\bigoplus_n\rbk{\nu_n^2+\physham[h]^2}$.
Each summand is nonnegative and self-adjoint; it is injective for $n\neq0$ because $\nu_n^2>0$ and for $n=0$ because
$\physham[h]$ is injective.
Hence $C=\bigoplus_n\rbk{\nu_n^2+\physham[h]^2}^{-1}$ and $C^{\frac{1}{2}}=\bigoplus_n\rbk{\nu_n^2+\physham[h]^2}^{-\frac{1}{2}}$.
For $n\neq0$, $\norm{\rbk{\nu_n^2+\physham[h]^2}^{-\frac{1}{2}}}\leq\abs{\nu_n}^{-1}\leq\nu_1^{-1}$, so only the zero mode
restricts the domain, which gives \eqref{eq:covariance-square-root}.
The norm of the zero-mode summand $\physham[h]^{-2}$ is finite exactly when $\inf\opvarspec{\physham[h]}>0$, and the other summands
have norm at most $\nu_1^{-2}$.
\end{proof}

For the basic model of Example \ref{ex:basic-model}, \(C\) is the multiplication by \(\rbk{\nu_n^2+\abs{2\pi k}^{2s}}^{-1}\) in the variables \(\pairbk{n,k}\), which is unbounded near \(\pairbk{0,0}\). The injectivity of \(\physham[h]\) cannot be dropped: if \(A\geq0\) is a self-adjoint operator on \(\sphilbfrak{h}\) with \(\ker A\neq\setone{0}\), then every constant function with values in \(\ker A\) is annihilated by \(D_t^2+A^2\), so the covariance operator does not exist. Correspondingly, the Araki--Woods representations of \cite[Sections 17.1.4 and 17.1.5, pp.~433--437]{DerezinskiGerard001} assume \(\ker\rbk{1-\gamma}=\setone{0}\), which for \(\gamma=\napiernum^{-\sminvtemperature\physham[h]}\) is the injectivity of \(\physham[h]\).

For \(\tau\in\fldreal\), the circle translation \(u_\tau\) and the reflection \(r\) are the orthogonal operators \begin{equation}\label{eq:circle-translation-reflection}
\begin{aligned}
\fun{\rbk{u_\tau F}}{t}
&=
\fun{F}{t-\tau},
&
\fun{\rbk{rF}}{t}
&=
\fun{F}{-t}
\end{aligned}
\end{equation} on \(\fun{\lp^2}{S_{\sminvtemperature};\sphilbfrak{h}_{\fldreal}}\). They satisfy \(u_{\sigma+\tau}=u_\sigma u_\tau\), \(u_{\sminvtemperature}=1\), \(r^2=1\), and \(ru_\tau=u_{-\tau}r\). Since \(\rbk{u_\tau F}_n=\napiernum^{-\imunit\nu_n\tau}F_n\) and \(\rbk{rF}_n=F_{-n}\), Lemma \ref{lem:covariance-fourier} shows that they commute with \(C\) and \(C^{\frac{1}{2}}\).

\section{Sharp-time vectors and half-circle projections}\label{sec:sharp-time}

In this section, we construct the covariance Hilbert space, the sharp-time vectors, and the test space, and we prove the half-circle projection formula, which is the one-particle form of the Markov property.

\subsection{Covariance Hilbert space and sharp-time vectors}\label{covariance-hilbert-space-and-sharp-time-vectors}

In this subsection, we define the sharp-time vectors in the completion of \(\dom C\) and compute their inner products.

\begin{defn}[Covariance Hilbert space]\label{def:covariance-hilbert-space}
The covariance Hilbert space $\sphilbfrak{H}_{\sminvtemperature,C}$ is the completion of $\dom C$ in the inner product
$\bkt{f}{g}_{\sphilbfrak{H}_{\sminvtemperature,C}}=\bkt{C^{\frac{1}{2}}f}{C^{\frac{1}{2}}g}_{\fun{\lp^2}{S_{\sminvtemperature};\sphilbfrak{h}_{\fldreal}}}$.
Let $V\colon\sphilbfrak{H}_{\sminvtemperature,C}\to\fun{\lp^2}{S_{\sminvtemperature};\sphilbfrak{h}_{\fldreal}}$ be the continuous extension of
$C^{\frac{1}{2}}$ restricted to $\dom C$.
\end{defn}

The map \(V\) is unitary, because \(C\) is injective and the range of \(C^{\frac{1}{2}}\) restricted to \(\dom C\) contains the ranges of the spectral projections \(\fun{\msrprb_C}{\closedinterval{N^{-1}}{N}}\), \(N\in\semigrposint\); in particular, \(\sphilbfrak{H}_{\sminvtemperature,C}\) is separable. The operators \(u_\tau\) and \(r\) preserve \(\dom C\) and the inner product of \(\sphilbfrak{H}_{\sminvtemperature,C}\); their continuous extensions, denoted by the same symbols, are orthogonal operators on \(\sphilbfrak{H}_{\sminvtemperature,C}\) with the relations after \eqref{eq:circle-translation-reflection}, and \(Vu_\tau=u_\tau V\) and \(Vr=rV\). In particular, \(\fml{u_\tau}{\tau\in\fldreal}\) is strongly continuous on \(\sphilbfrak{H}_{\sminvtemperature,C}\).

The sharp-time vectors are controlled by the following Matsubara sum. For \(u\in\closedinterval{0}{\sminvtemperature}\) and \(\lambda>0\), set \begin{equation}\label{eq:green-multiplier}
\begin{aligned}
\fun{M_u}{\lambda}
&=
\frac{\napiernum^{-u\lambda}+\napiernum^{-\rbk{\sminvtemperature-u}\lambda}}{\lambda\rbk{1-\napiernum^{-\sminvtemperature\lambda}}},
&
\fun{\varpi}{\lambda}
&=
\rbk{\frac{1}{2\lambda}\fun{\coth}{\frac{\sminvtemperature\lambda}{2}}}^{\frac{1}{2}}
=
\rbk{\frac{1}{2}\fun{M_0}{\lambda}}^{\frac{1}{2}}.
\end{aligned}
\end{equation}

\begin{lem}[Matsubara sum]\label{lem:matsubara-sum}
Let $u\in\closedinterval{0}{\sminvtemperature}$ and $g_1,g_2\in\dom\physham[h]^{-1}$.
Then $\dom\fun{\varpi}{\physham[h]}=\dom\physham[h]^{-1}$, the series and the quadratic form of the functional calculus in
\begin{equation}\label{eq:matsubara-sum}
\begin{aligned}
\frac{1}{\sminvtemperature}\sum_{n\in\ringratint}\napiernum^{\imunit\nu_nu}
\bkt{\rbk{\nu_n^2+\physham[h]^2}^{-\frac{1}{2}}g_1}{\rbk{\nu_n^2+\physham[h]^2}^{-\frac{1}{2}}g_2}_{\sphilbfrak{h}_{\fldreal}}
&=
\frac{1}{2}\bkt{g_1}{\fun{M_u}{\physham[h]}g_2}_{\sphilbfrak{h}_{\fldreal}}
\end{aligned}
\end{equation}
converge absolutely, and the identity holds.
In particular,
$$\begin{aligned}
\sum_{n\in\ringratint}\norm{\rbk{\nu_n^2+\physham[h]^2}^{-\frac{1}{2}}g_1}^2
&=
\sminvtemperature\norm{\fun{\varpi}{\physham[h]}g_1}^2
<
\infty.
\end{aligned}$$
\end{lem}

\begin{proof}
For a scalar $a>0$, the function $u\mapsto\frac{1}{2}\fun{M_u}{a}$ on $\closedinterval{0}{\sminvtemperature}$ has continuous
$\sminvtemperature$-periodic extension and Fourier coefficients $\sminvtemperature^{-\frac{1}{2}}\rbk{\nu_n^2+a^2}^{-1}$, so a direct
computation gives, with uniform convergence,
$$\begin{aligned}
\frac{1}{\sminvtemperature}\sum_{n\in\ringratint}\frac{\napiernum^{\imunit\nu_nu}}{\nu_n^2+a^2}
&=
\frac{1}{2}\fun{M_u}{a};
\end{aligned}$$
see \cite[(21.53), p.~634]{DerezinskiGerard001}.
The estimates $\fun{M_u}{\lambda}\leq\fun{M_0}{\lambda}=\lambda^{-1}\fun{\coth}{\frac{\sminvtemperature\lambda}{2}}\leq1+\rbk{1+\frac{2}{\sminvtemperature}}\lambda^{-2}$
and $\fun{M_0}{\lambda}\geq\frac{2}{\sminvtemperature}\lambda^{-2}$ on $\lambda\leq1$ give $\dom\fun{\varpi}{\physham[h]}=\dom\physham[h]^{-1}$ and the
absolute convergence of the right-hand side of \eqref{eq:matsubara-sum}, by the Cauchy--Schwarz inequality for the
measures $\bkt{g_i}{\fun{\msrprb_{\physham[h]}}{\cdot}g_j}_{\sphilbfrak{h}_{\fldreal}}$, $i,j\in\setone{1,2}$.
The terms with $n\neq0$ are bounded by $\nu_n^{-2}$ uniformly in $a$, so the scalar identity can be integrated term by term
against the spectral measure $\msrprb_{\physham[h]}$, the term $n=0$ being $\sminvtemperature^{-1}\bkt{\physham[h]^{-1}g_1}{\physham[h]^{-1}g_2}$.
\end{proof}

\begin{defn}[Sharp-time vectors and the sharp-time space]\label{def:sharp-time-vectors}
For $t\in S_{\sminvtemperature}$ and $g\in\dom\physham[h]^{-1}$, the sharp-time vector $j_tg\in\sphilbfrak{H}_{\sminvtemperature,C}$ is
determined by
\begin{equation}\label{eq:sharp-time-vector}
\begin{aligned}
\rbk{Vj_tg}_n
&=
\frac{1}{\sqrt{\sminvtemperature}}\napiernum^{-\imunit\nu_nt}\rbk{\nu_n^2+\physham[h]^2}^{-\frac{1}{2}}g,
\quad n\in\ringratint,
\end{aligned}
\end{equation}
which is square summable by Lemma \ref{lem:matsubara-sum}.
Let $\sphilbfrak{R}$ be the completion of $\dom\physham[h]^{-1}$ in the norm $\norm{g}_{\sphilbfrak{R}}=\norm{\fun{\varpi}{\physham[h]}g}_{\sphilbfrak{h}_{\fldreal}}$.
For $s,t\in S_{\sminvtemperature}$ and $g_1,g_2\in\dom\physham[h]^{-1}$, the sharp-time Green kernel is
\begin{equation}\label{eq:green-kernel}
\begin{aligned}
\fun{\opform{G}_{\sminvtemperature}}{s,g_1;t,g_2}
&=
\frac{1}{2}\bkt{g_1}{\fun{M_{\fun{d_{\sminvtemperature}}{s,t}}}{\physham[h]}g_2}_{\sphilbfrak{h}_{\fldreal}}
\\ 
&=
\frac{1}{2}\bkt{g_1}{\physham[h]^{-1}\frac{\napiernum^{-\fun{d_{\sminvtemperature}}{s,t}\physham[h]}+\napiernum^{-\rbk{\sminvtemperature-\fun{d_{\sminvtemperature}}{s,t}}\physham[h]}}{1-\napiernum^{-\sminvtemperature\physham[h]}}g_2}_{\sphilbfrak{h}_{\fldreal}},
\end{aligned}
\end{equation}
where the right-hand sides are quadratic forms of the functional calculus.
\end{defn}

For a real \(\varphi\in\fun{\lp^2}{S_{\sminvtemperature};\fldreal}\), the vector \(V^{-1}C^{\frac{1}{2}}\rbk{\varphi\otimes g}\) has the coefficients \(\varphi_n\rbk{\nu_n^2+\physham[h]^2}^{-\frac{1}{2}}g\), so \(j_tg\) is the vector \(\msrdirac_{t}\otimes g\) for the Dirac measure \(\msrdirac_{t}\) at \(t\), whose coefficients \(\int_{S_{\sminvtemperature}}\sminvtemperature^{-\frac{1}{2}}\napiernum^{-\imunit\nu_nu}\opdmsr{\msrdirac_{t}(u)}=\sminvtemperature^{-\frac{1}{2}}\napiernum^{-\imunit\nu_nt}\) replace those of \(\varphi\) in \eqref{eq:fourier-coefficient}. Since \(\fun{\varpi}{\lambda}=\rbk{2\lambda\fun{\tanh}{\frac{\sminvtemperature\lambda}{2}}}^{-\frac{1}{2}}\), the completion \(\sphilbfrak{R}\) of \(\pairbk{\dom\physham[h]^{-1},\norm{\cdot}_{\sphilbfrak{R}}}\) is the space \(\rbk{2\physham[h]\fun{\tanh}{\frac{\sminvtemperature\physham[h]}{2}}}^{\frac{1}{2}}\sphilbfrak{h}_{\fldreal}\) of \cite[Definition 2.59, p.~47]{DerezinskiGerard001}, on which the sharp-time maps of \cite[Lemma 21.51, p.~633]{DerezinskiGerard001} are defined.

\begin{prop}[Sharp-time vectors]\label{prop:sharp-time-vectors}
Let $s,t\in S_{\sminvtemperature}$, $\tau\in\fldreal$, and $g_1,g_2\in\dom\physham[h]^{-1}$.
Then
\begin{equation}\label{eq:sharp-time-pairing}
\begin{aligned}
\bkt{j_sg_1}{j_tg_2}_{\sphilbfrak{H}_{\sminvtemperature,C}}
&=
\fun{\opform{G}_{\sminvtemperature}}{s,g_1;t,g_2},
&
u_\tau j_tg_1
&=
j_{t+\tau}g_1,
&
rj_tg_1
&=
j_{-t}g_1.
\end{aligned}
\end{equation}
In particular, $j_t$ extends to an isometry from $\sphilbfrak{R}$ into $\sphilbfrak{H}_{\sminvtemperature,C}$, and these relations hold for
$g_1,g_2\in\sphilbfrak{R}$ with the continuous extension of $\opform{G}_{\sminvtemperature}$.
\end{prop}

\begin{proof}
By \eqref{eq:real-parseval} and \eqref{eq:sharp-time-vector}, the left-hand side of the first identity is
$\sminvtemperature^{-1}\sum_n\opreal\rbk{\napiernum^{\imunit\nu_n\rbk{s-t}}w_n}$ with
$w_n=\bkt{\rbk{\nu_n^2+\physham[h]^2}^{-\frac{1}{2}}g_1}{\rbk{\nu_n^2+\physham[h]^2}^{-\frac{1}{2}}g_2}_{\sphilbfrak{h}}=w_{-n}$.
Since $\nu_{-n}=-\nu_n$, grouping the terms $n$ and $-n$ shows that this sum equals
$\sminvtemperature^{-1}\sum_n\napiernum^{\imunit\nu_n\rbk{s-t}}\opreal w_n$, and $\opreal w_n$ is the inner product of $\sphilbfrak{h}_{\fldreal}$.
Lemma \ref{lem:matsubara-sum} with $u\in\closedinterval{0}{\sminvtemperature}$ congruent to $s-t$, together with $M_u=M_{\sminvtemperature-u}$, gives
\eqref{eq:green-kernel}.
At $s=t$ this is $\norm{j_tg_1}^2=\norm{g_1}_{\sphilbfrak{R}}^2$.
The last two identities follow from $\rbk{u_\tau F}_n=\napiernum^{-\imunit\nu_n\tau}F_n$, $\rbk{rF}_n=F_{-n}$, and $\nu_{-n}=-\nu_n$.
\end{proof}

\subsection{Test spaces}\label{test-spaces}

In this subsection, we introduce the test spaces that generate the time-zero fields and prove that their sharp-time vectors span the covariance Hilbert space.

For \(\sigma\in\closedinterval{-\frac{\sminvtemperature}{2}}{0}\) and \(\lambda>0\), set \begin{equation}\label{eq:sinh-ratios}
\begin{aligned}
\fun{\eta_{0,\sigma}}{\lambda}
&=
\frac{\fun{\sinh}{\rbk{\frac{\sminvtemperature}{2}+\sigma}\lambda}}{\fun{\sinh}{\frac{\sminvtemperature\lambda}{2}}},
&
\fun{\eta_{1,\sigma}}{\lambda}
&=
\frac{\fun{\sinh}{-\sigma\lambda}}{\fun{\sinh}{\frac{\sminvtemperature\lambda}{2}}}.
\end{aligned}
\end{equation} These functions take values in \(\closedinterval{0}{1}\), so \(\fun{\eta_{0,\sigma}}{\physham[h]}\) and \(\fun{\eta_{1,\sigma}}{\physham[h]}\) are non-expansive, that is, have operator norm at most \(1\).

\begin{defn}[Test space]\label{def:test-space}
A test space is a dense real linear subspace $\mathcal{T}_{\fldreal}$ of $\sphilbfrak{h}_{\fldreal}$ with
$\mathcal{T}_{\fldreal}\subset\dom\physham[h]^{-1}$ and
\begin{equation}\label{eq:test-space-invariance}
\begin{aligned}
\fun{\eta_{0,\sigma}}{\physham[h]}\mathcal{T}_{\fldreal}
&\subset
\mathcal{T}_{\fldreal},
&
\fun{\eta_{1,\sigma}}{\physham[h]}\mathcal{T}_{\fldreal}
&\subset
\mathcal{T}_{\fldreal},
\quad
\sigma\in\closedinterval{-\frac{\sminvtemperature}{2}}{0}.
\end{aligned}
\end{equation}
In what follows, a test space $\mathcal{T}_{\fldreal}$ is fixed.
\end{defn}

\begin{ex}[Test spaces]\label{ex:test-spaces}
For the basic model of Example \ref{ex:basic-model}, the space of $f$ with $\fafouriertr f\in\fun{C_{\mathrm{c}}^{\infty}}{\fldreal^d\setminus\setone{0}}$
is a test space, because the functions $k\mapsto\fun{\eta_{i,\sigma}}{\abs{2\pi k}^s}$ are smooth on $\fldreal^d\setminus\setone{0}$.
For a general $\physham[h]$, the union over $N\in\semigrposint$ of the ranges of the spectral projections
$\fun{\msrprb_{\physham[h]}}{\closedinterval{N^{-1}}{N}}$, regarded as real subspaces of $\sphilbfrak{h}_{\fldreal}$, is a test space.
\end{ex}

\begin{lem}[Weighted density]\label{lem:weighted-density}
Let $\varphi\colon\openinterval{0}{\infty}\to\openinterval{0}{\infty}$ be a Borel function that is bounded on compact subsets, and
assume $\mathcal{T}_{\fldreal}\subset\dom\fun{\varphi}{\physham[h]}$.
Then $\fun{\varphi}{\physham[h]}\mathcal{T}_{\fldreal}$ is dense in $\sphilbfrak{h}_{\fldreal}$.
\end{lem}

\begin{proof}
At $\sigma=-\frac{\sminvtemperature}{4}$, both functions \eqref{eq:sinh-ratios} equal
$\fun{\kappa}{\lambda}=\rbk{2\fun{\cosh}{\frac{\sminvtemperature\lambda}{4}}}^{-1}$, a strictly decreasing bijection of
$\openinterval{0}{\infty}$ onto $\openinterval{0}{\frac{1}{2}}$.
Hence $\fun{\kappa}{\physham[h]}$ and every real polynomial $\fun{\rbk{p\circ\kappa}}{\physham[h]}$ preserve $\mathcal{T}_{\fldreal}$, and these
operators commute with $\fun{\varphi}{\physham[h]}$ on $\mathcal{T}_{\fldreal}$.
For $0<\varepsilon<M$, the spectral projection $\fun{\msrprb_{\physham[h]}}{\closedinterval{\varepsilon}{M}}$ is $\fun{\rbk{\chi\circ\kappa}}{\physham[h]}$ for
the indicator function $\chi$ of a compact subinterval of $\openinterval{0}{\frac{1}{2}}$, and the Weierstrass theorem gives real
polynomials $p_m$ with $\sup_m\norm{\fun{\rbk{p_m\circ\kappa}}{\physham[h]}}<\infty$ and
$\fun{\rbk{p_m\circ\kappa}}{\physham[h]}\to\fun{\msrprb_{\physham[h]}}{\closedinterval{\varepsilon}{M}}$ strongly.
If $z\in\sphilbfrak{h}_{\fldreal}$ is orthogonal to $\fun{\varphi}{\physham[h]}\mathcal{T}_{\fldreal}$, then so is
$z'=\fun{\msrprb_{\physham[h]}}{\closedinterval{\varepsilon}{M}}z$; since $\fun{\varphi}{\physham[h]}\fun{\msrprb_{\physham[h]}}{\closedinterval{\varepsilon}{M}}$ is bounded and
self-adjoint, $\fun{\varphi}{\physham[h]}z'$ is orthogonal to the dense space $\mathcal{T}_{\fldreal}$, so $z'=0$ because $\varphi>0$.
Letting $\varepsilon\downarrow0$ and $M\uparrow\infty$ gives $z=0$.
\end{proof}

\begin{lem}[Totality of sharp-time vectors]\label{lem:sharp-time-total}
The vectors $j_tf$ with $t\in S_{\sminvtemperature}$ and $f\in\mathcal{T}_{\fldreal}$ have dense real linear span in
$\sphilbfrak{H}_{\sminvtemperature,C}$, and $\mathcal{T}_{\fldreal}$ is dense in $\sphilbfrak{R}$.
\end{lem}

\begin{proof}
Let $w$ be orthogonal to every $j_tf$, let $W_n=\rbk{Vw}_n$, and let $B_n=\rbk{\nu_n^2+\physham[h]^2}^{-\frac{1}{2}}$.
For $f\in\mathcal{T}_{\fldreal}$, put $\gamma_n=\bkt{W_n}{B_nf}_{\sphilbfrak{h}}$, which is summable by Lemma \ref{lem:matsubara-sum}.
By \eqref{eq:real-parseval} and \eqref{eq:sharp-time-vector},
$$\begin{aligned}
0
=
\bkt{w}{j_tf}_{\sphilbfrak{H}_{\sminvtemperature,C}}
&=
\frac{1}{2\sqrt{\sminvtemperature}}\sum_{n\in\ringratint}\napiernum^{-\imunit\nu_nt}\rbk{\gamma_n+\cmpconj{\gamma_{-n}}},
\quad t\in S_{\sminvtemperature},
\end{aligned}$$
so $\gamma_n+\cmpconj{\gamma_{-n}}=0$.
Since $B_{-n}=B_n$, the real and imaginary parts give that $W_n+W_{-n}$ and $\imunit\rbk{W_n-W_{-n}}$ are orthogonal to
$B_n\mathcal{T}_{\fldreal}$ in $\sphilbfrak{h}_{\fldreal}$.
Lemma \ref{lem:weighted-density} with $\fun{\varphi}{\lambda}=\rbk{\nu_n^2+\lambda^2}^{-\frac{1}{2}}$ shows that $B_n\mathcal{T}_{\fldreal}$ is
dense, so $W_n=0$ for every $n$ and $w=0$.
The density of $\mathcal{T}_{\fldreal}$ in $\sphilbfrak{R}$ is Lemma \ref{lem:weighted-density} with $\varphi=\varpi$.
\end{proof}

\subsection{Half-circle projections}\label{half-circle-projections}

In this subsection, we prove the projection formula that gives the Markov property.

For \(I\subset S_{\sminvtemperature}\), let \(\sphilbfrak{k}_{\sminvtemperature,I}\) be the closed real linear span of the vectors \(j_tf\) with \(t\in I\) and \(f\in\mathcal{T}_{\fldreal}\), and let \(e_I\) be the orthogonal projection of \(\sphilbfrak{H}_{\sminvtemperature,C}\) onto \(\sphilbfrak{k}_{\sminvtemperature,I}\). Let \begin{equation}\label{eq:half-circles}
\begin{aligned}
I_+
&=
\closedinterval{0}{\frac{\sminvtemperature}{2}},
&
I_-
&=
\closedinterval{-\frac{\sminvtemperature}{2}}{0},
&
\partial
&=
\setone{0,\frac{\sminvtemperature}{2}}.
\end{aligned}
\end{equation}

\begin{prop}[Half-circle projections]\label{prop:half-circle-projections}
For $s\in I_-$ and $f\in\mathcal{T}_{\fldreal}$,
\begin{equation}\label{eq:half-circle-projection}
\begin{aligned}
e_{I_+}j_sf
&=
j_0\fun{\eta_{0,s}}{\physham[h]}f+j_{\frac{\sminvtemperature}{2}}\fun{\eta_{1,s}}{\physham[h]}f,
\end{aligned}
\end{equation}
and
\begin{equation}\label{eq:half-circle-projection-product}
\begin{aligned}
e_{I_+}e_{I_-}
&=
e_{\partial}.
\end{aligned}
\end{equation}
\end{prop}

\begin{proof}
For $\lambda>0$ and $s\in I_-$, both sides of
$$\begin{aligned}
\fun{\cosh}{\rbk{\frac{\sminvtemperature}{2}+s-t}\lambda}
&=
\fun{\eta_{0,s}}{\lambda}\fun{\cosh}{\rbk{\frac{\sminvtemperature}{2}-t}\lambda}
+\fun{\eta_{1,s}}{\lambda}\fun{\cosh}{t\lambda}
\end{aligned}$$
solve $y''=\lambda^2y$ in $t$, and the addition formulas show that they agree at $t=0$ and $t=\frac{\sminvtemperature}{2}$; hence the
identity holds for $t\in I_+$.
For $t\in I_+$ and $s\in I_-$, one has $0\leq t-s\leq\sminvtemperature$, and \eqref{eq:green-kernel} reads
$$\begin{aligned}
\fun{\opform{G}_{\sminvtemperature}}{s,f;t,g}
&=
\bkt{f}{\frac{\fun{\cosh}{\rbk{\frac{\sminvtemperature}{2}-\rbk{t-s}}\physham[h]}}{2\physham[h]\fun{\sinh}{\frac{\sminvtemperature\physham[h]}{2}}}g}_{\sphilbfrak{h}_{\fldreal}}.
\end{aligned}$$
Dividing the scalar identity by $2\lambda\fun{\sinh}{\frac{\sminvtemperature\lambda}{2}}$ and integrating against the spectral measure $\msrprb_{\physham[h]}$ gives,
by Proposition \ref{prop:sharp-time-vectors},
$$\begin{aligned}
\bkt{j_sf}{j_tg}_{\sphilbfrak{H}_{\sminvtemperature,C}}
&=
\bkt{j_0\fun{\eta_{0,s}}{\physham[h]}f+j_{\frac{\sminvtemperature}{2}}\fun{\eta_{1,s}}{\physham[h]}f}{j_tg}_{\sphilbfrak{H}_{\sminvtemperature,C}},
\quad t\in I_+,\ g\in\mathcal{T}_{\fldreal}.
\end{aligned}$$
By \eqref{eq:test-space-invariance}, the vector on the right-hand side belongs to $\sphilbfrak{k}_{\sminvtemperature,\partial}\subset\sphilbfrak{k}_{\sminvtemperature,I_+}$,
which proves \eqref{eq:half-circle-projection}.

By \eqref{eq:half-circle-projection}, $e_{I_+}$ maps $\sphilbfrak{k}_{\sminvtemperature,I_-}$ into $\sphilbfrak{k}_{\sminvtemperature,\partial}$.
Hence, for $v\in\sphilbfrak{H}_{\sminvtemperature,C}$, both $e_{I_+}e_{I_-}v$ and $e_{\partial}v$ lie in $\sphilbfrak{k}_{\sminvtemperature,\partial}$, and for
$w\in\sphilbfrak{k}_{\sminvtemperature,\partial}\subset\sphilbfrak{k}_{\sminvtemperature,I_+}\cap\sphilbfrak{k}_{\sminvtemperature,I_-}$,
$\bkt{e_{I_+}e_{I_-}v}{w}=\bkt{v}{e_{I_-}e_{I_+}w}=\bkt{v}{w}=\bkt{e_{\partial}v}{w}$.
\end{proof}

Proposition \ref{prop:half-circle-projections} corresponds to the formula for the projections onto the closed intervals of \(S_{\sminvtemperature}\) and the \(\sminvtemperature\)-Markov property in \cite[Proposition 21.54(1) and (7), p.~634]{DerezinskiGerard001}. Its proof uses only the non-expansive operators \(\fun{\eta_{i,s}}{\physham[h]}\) and the test space, so no spectral gap enters.

\section{\texorpdfstring{The Gaussian measure and the \(\sminvtemperature\)-Markov property}{The Gaussian measure and the \textbackslash sminvtemperature-Markov property}}\label{sec:markov}

In this section, we realize the Gaussian field with covariance \(\frac{1}{2}\bkt{\cdot}{\cdot}_{\sphilbfrak{H}_{\sminvtemperature,C}}\) by a Borel measure on a Hilbert space on which the circle translations and the reflection act, and we prove that the resulting path space is a \(\sminvtemperature\)-Markov path space.

\subsection{The regularized Gaussian measure}\label{the-regularized-gaussian-measure}

In this subsection, we construct the Hilbert scale and the Gaussian measure.

\begin{defn}[Invariant regularizer]\label{def:invariant-regularizer}
An invariant regularizer is a self-adjoint operator $\opform{R}_{\sminvtemperature,C}\geq1$ on $\sphilbfrak{H}_{\sminvtemperature,C}$ such that
$\opform{R}_{\sminvtemperature,C}^{-1}$ is Hilbert--Schmidt and $\opform{R}_{\sminvtemperature,C}$ commutes with $u_\tau$, $\tau\in\fldreal$, and $r$.
For such an operator, let $\prbqspace_+=\dom\opform{R}_{\sminvtemperature,C}$ with the norm $\norm{\opform{R}_{\sminvtemperature,C}\,\cdot}$, and let
$\prbqspace_-$ be the completion of $\sphilbfrak{H}_{\sminvtemperature,C}$ in the norm $\norm{\opform{R}_{\sminvtemperature,C}^{-1}\,\cdot}$.
The pairing $\bkt{\cdot}{\cdot}_{\prbqspace_-,\prbqspace_+}$ is the continuous extension of the inner product of
$\sphilbfrak{H}_{\sminvtemperature,C}$ to $\prbqspace_-\times\prbqspace_+$.
In what follows, an invariant regularizer is fixed.
\end{defn}

\begin{lem}[Existence of an invariant regularizer]\label{lem:regularizer-existence}
Let $\fml{\psi_m}{m\in\mathcal{M}}$ be a complete orthonormal system of $\sphilbfrak{h}_{\fldreal}$ with a countable index set $\mathcal{M}$, let
$c_0=\sminvtemperature^{-\frac{1}{2}}$, $\fun{c_n}{t}=\rbk{\frac{2}{\sminvtemperature}}^{\frac{1}{2}}\fun{\cos}{\nu_nt}$, and
$\fun{s_n}{t}=\rbk{\frac{2}{\sminvtemperature}}^{\frac{1}{2}}\fun{\sin}{\nu_nt}$ for $n\in\semigrposint$, and let
$\semigrposint\ni j\mapsto\pairbk{n_j,m_j}\in\monnat\times\mathcal{M}$ be a bijection.
Let $\sphilbfrak{V}_j$ be $V^{-1}$ of the span of $c_0\otimes\psi_{m_j}$ if $n_j=0$ and of the span of
$c_{n_j}\otimes\psi_{m_j}$ and $s_{n_j}\otimes\psi_{m_j}$ if $n_j\geq1$, let $P_j$ be the orthogonal projection onto $\sphilbfrak{V}_j$, and let
$d_j=\dim\sphilbfrak{V}_j$.
Then $\opform{R}_{\sminvtemperature,C}=\sum_{j\in\semigrposint}2^j\sqrt{d_j}P_j$ is an invariant regularizer with
$\norm{\opform{R}_{\sminvtemperature,C}^{-1}}_{\txthilbertschmidt}^2=\frac{1}{3}$.
\end{lem}

\begin{proof}
The products of the real trigonometric system and $\fml{\psi_m}{m}$ form a complete orthonormal system of
$\fun{\lp^2}{S_{\sminvtemperature};\sphilbfrak{h}_{\fldreal}}$, so $\sphilbfrak{H}_{\sminvtemperature,C}=\bigoplus_j\sphilbfrak{V}_j$.
Each $\sphilbfrak{V}_j$ is invariant under $u_\tau$, which acts as a rotation of $\pairbk{c_n\otimes\psi_m,s_n\otimes\psi_m}$, and under $r$,
which fixes $c_n\otimes\psi_m$ and changes the sign of $s_n\otimes\psi_m$; hence $P_j$ commutes with $u_\tau$ and $r$.
Finally, $\sum_jd_j\rbk{2^j\sqrt{d_j}}^{-2}=\sum_j4^{-j}=\frac{1}{3}$.
\end{proof}

The operators \(u_\tau\) and \(r\) preserve \(\prbqspace_+\) and the norms of \(\prbqspace_\pm\). Their continuous extensions \(\tilde{u}_\tau\) and \(\tilde{r}\) to \(\prbqspace_-\) are orthogonal operators with the relations after \eqref{eq:circle-translation-reflection}, and \begin{equation}\label{eq:dual-symmetries}
\begin{aligned}
\bkt{\tilde{u}_{-\tau}q}{v}_{\prbqspace_-,\prbqspace_+}
&=
\bkt{q}{u_\tau v}_{\prbqspace_-,\prbqspace_+},
&
\bkt{\tilde{r}q}{v}_{\prbqspace_-,\prbqspace_+}
&=
\bkt{q}{rv}_{\prbqspace_-,\prbqspace_+},
\quad q\in\prbqspace_-,\ v\in\prbqspace_+.
\end{aligned}
\end{equation}

For \(f\in\prbqspace_+\), let \(\fun{\opfocksegal}{f}\) be the continuous linear function \(q\mapsto\bkt{q}{f}_{\prbqspace_-,\prbqspace_+}\) on \(\prbqspace_-\).

\begin{prop}[Gaussian measure]\label{prop:gaussian-measure}
There is a unique Borel probability measure $\msrprb_{\sminvtemperature,0}$ on $\prbqspace_-$ such that
\begin{equation}\label{eq:characteristic-functional}
\begin{aligned}
\sqfun{\prbexp_{\msrprb_{\sminvtemperature,0}}}{\napiernum^{\imunit\fun{\opfocksegal}{f}}}
&=
\fnexp{-\frac{1}{4}\norm{f}_{\sphilbfrak{H}_{\sminvtemperature,C}}^2},
\quad f\in\prbqspace_+.
\end{aligned}
\end{equation}
The map $\opfocksegal$ extends uniquely to a continuous linear map
$\opfocksegal\colon\sphilbfrak{H}_{\sminvtemperature,C}\to\fun{\lp^2}{\prbqspace_-,\msrprb_{\sminvtemperature,0}}$ with
$\sqfun{\prbexp_{\msrprb_{\sminvtemperature,0}}}{\fun{\opfocksegal}{v}\fun{\opfocksegal}{w}}=\frac{1}{2}\bkt{v}{w}_{\sphilbfrak{H}_{\sminvtemperature,C}}$,
every $\fun{\opfocksegal}{v}$ is centered Gaussian, and $\msrprb_{\sminvtemperature,0}$ is invariant under $\tilde{u}_\tau$, $\tau\in\fldreal$, and $\tilde{r}$.
\end{prop}

\begin{proof}
The existence and the extension of $\opfocksegal$ are \cite[Lemma 1.119 and Theorem 1.120, pp.~95--97]{LorincziHiroshimaBetz003}
with $L=\frac{1}{2}$ and the Hilbert--Schmidt operator $\opform{R}_{\sminvtemperature,C}^{-1}$.
The uniqueness holds because the functions $\fun{\opfocksegal}{f}$, $f\in\prbqspace_+$, are all the continuous linear functionals of the
separable Hilbert space $\prbqspace_-$.
By \eqref{eq:dual-symmetries}, the images of $\msrprb_{\sminvtemperature,0}$ under $\tilde{u}_{-\tau}$ and $\tilde{r}$ have the characteristic
functional \eqref{eq:characteristic-functional}, so they equal $\msrprb_{\sminvtemperature,0}$.
\end{proof}

If \(\fml{\epsilon_k}{k\in\semigrposint}\) is a complete orthonormal system of \(\sphilbfrak{H}_{\sminvtemperature,C}\) consisting of eigenvectors of \(\opform{R}_{\sminvtemperature,C}\) and \(\fml{\xi_k}{k\in\semigrposint}\) are independent standard Gaussian random variables, then \(2^{-\frac{1}{2}}\sum_k\xi_k\epsilon_k\) converges almost surely in \(\prbqspace_-\), since \(\sum_k\norm{\epsilon_k}_{\prbqspace_-}^2=\norm{\opform{R}_{\sminvtemperature,C}^{-1}}_{\txthilbertschmidt}^2\), and its law is \(\msrprb_{\sminvtemperature,0}\). Since the variance of \(\fun{\opfocksegal}{f}\) is \(\frac{1}{2}\norm{f}_{\sphilbfrak{H}_{\sminvtemperature,C}}^2\), the Cameron--Martin space of \(\msrprb_{\sminvtemperature,0}\) in the sense of \cite[Section 2.4]{VladimirBogachev002} is \(\sphilbfrak{H}_{\sminvtemperature,C}\subset\prbqspace_-\) with the norm \(\sqrt{2}\norm{\cdot}_{\sphilbfrak{H}_{\sminvtemperature,C}}\).

\subsection{\texorpdfstring{The singular Gaussian \(\sminvtemperature\)-Markov path space}{The singular Gaussian \textbackslash sminvtemperature-Markov path space}}\label{the-singular-gaussian-sminvtemperature-markov-path-space}

In this subsection, we define the singular Gaussian \(\sminvtemperature\)-Markov path space and prove the \(\sminvtemperature\)-Markov property.

\begin{defn}[Singular Gaussian $\sminvtemperature$-Markov path space]\label{def:singular-gaussian-beta-markov-path-space}
For $t\in S_{\sminvtemperature}$ and $g\in\sphilbfrak{R}$, let
$\fun{\opfocksegal_t}{g}=\fun{\opfocksegal}{j_tg}\in\fun{\lp^2}{\prbqspace_-,\msrprb_{\sminvtemperature,0}}$ be the sharp-time field.
Let $\mblfmlfrak{S}$ be the $\msrprb_{\sminvtemperature,0}$-completion of the Borel $\sigma$-algebra of $\prbqspace_-$, and for
$I\subset S_{\sminvtemperature}$ let $\mblfmlfrak{S}_I$ be generated by $\fun{\opfocksegal_t}{f}$, $t\in I$, $f\in\mathcal{T}_{\fldreal}$, and the
null sets; set $\mblfmlfrak{S}_0=\mblfmlfrak{S}_{\setone{0}}$, and let $E_I$ be the conditional expectation onto $\mblfmlfrak{S}_I$.
For a function $F$ on $\prbqspace_-$ and $\tau\in\fldreal$, set
\begin{equation}\label{eq:path-space-automorphisms}
\begin{aligned}
U_\tau F
&=
F\circ\tilde{u}_{-\tau},
&
RF
&=
F\circ\tilde{r}.
\end{aligned}
\end{equation}
The singular Gaussian $\sminvtemperature$-Markov path space is the sextuple
$$\begin{aligned}
\rbk{\prbqspace_-,\mblfmlfrak{S},\mblfmlfrak{S}_0,\fml{U_t}{t\in\fldreal},R,\msrprb_{\sminvtemperature,0}}.
\end{aligned}$$
\end{defn}

By Propositions \ref{prop:sharp-time-vectors} and \ref{prop:gaussian-measure}, the sharp-time fields are centered Gaussian with \begin{equation}\label{eq:sharp-time-field-covariance}
\begin{aligned}
\sqfun{\prbexp_{\msrprb_{\sminvtemperature,0}}}{\fun{\opfocksegal_s}{f}\fun{\opfocksegal_t}{g}}
&=
\frac{1}{2}\fun{\opform{G}_{\sminvtemperature}}{s,f;t,g},
\quad s,t\in S_{\sminvtemperature},\ f,g\in\dom\physham[h]^{-1}.
\end{aligned}
\end{equation}

For a non-expansive operator \(A\) on \(\sphilbfrak{H}_{\sminvtemperature,C}\), the second quantization \(\fun{\Gamma}{A}\) of \(A\) in the real-wave representation \cite[Section 9.3]{DerezinskiGerard001} is the non-expansive operator on \(\fun{\lp^2}{\prbqspace_-,\msrprb_{\sminvtemperature,0}}\) with \begin{equation}\label{eq:gaussian-second-quantization}
\begin{aligned}
\fun{\Gamma}{A}\napiernum^{\imunit\fun{\opfocksegal}{v}}
&=
\fnexp{\frac{1}{4}\rbk{\norm{Av}^2-\norm{v}^2}}\napiernum^{\imunit\fun{\opfocksegal}{Av}},
\quad v\in\sphilbfrak{H}_{\sminvtemperature,C}.
\end{aligned}
\end{equation} It satisfies \(\fun{\Gamma}{A}\fun{\Gamma}{B}=\fun{\Gamma}{AB}\).

\begin{prop}[Symmetries and conditional expectations]\label{prop:symmetries-conditional-expectations}
For $\tau\in\fldreal$, the maps $U_\tau$ and $R$ are measure-preserving $\ast$-automorphisms of
$\fun{\lp^\infty}{\prbqspace_-,\mblfmlfrak{S},\msrprb_{\sminvtemperature,0}}$, and
\begin{equation}\label{eq:field-transformations}
\begin{aligned}
U_\tau\fun{\opfocksegal}{v}
&=
\fun{\opfocksegal}{u_\tau v},
&
R\fun{\opfocksegal}{v}
&=
\fun{\opfocksegal}{rv},
&
U_\tau\fun{\opfocksegal_t}{g}
&=
\fun{\opfocksegal_{t+\tau}}{g},
&
R\fun{\opfocksegal_t}{g}
&=
\fun{\opfocksegal_{-t}}{g}
\end{aligned}
\end{equation}
for $v\in\sphilbfrak{H}_{\sminvtemperature,C}$, $t\in S_{\sminvtemperature}$, and $g\in\sphilbfrak{R}$.
For $I\subset S_{\sminvtemperature}$, the $\sigma$-algebra $\mblfmlfrak{S}_I$ is generated by $\fun{\opfocksegal}{\sphilbfrak{k}_{\sminvtemperature,I}}$ and the
null sets, $\mblfmlfrak{S}_{S_{\sminvtemperature}}=\mblfmlfrak{S}$, and
\begin{equation}\label{eq:conditional-expectation-second-quantization}
\begin{aligned}
E_I
&=
\fun{\Gamma}{e_I},
&
K_{U_\tau}
&=
\fun{\Gamma}{u_\tau},
&
K_R
&=
\fun{\Gamma}{r},
\end{aligned}
\end{equation}
where $K_{U_\tau}$ and $K_R$ are the unitary operators on $\fun{\lp^2}{\prbqspace_-,\msrprb_{\sminvtemperature,0}}$ induced by $U_\tau$ and $R$.
\end{prop}

\begin{proof}
The first assertion is Proposition \ref{prop:gaussian-measure}.
The identities \eqref{eq:field-transformations} hold for $v\in\prbqspace_+$ by \eqref{eq:dual-symmetries}, extend to
$\sphilbfrak{H}_{\sminvtemperature,C}$ by continuity, and give the sharp-time relations by Proposition \ref{prop:sharp-time-vectors}.
Since $\fun{\opfocksegal}{v}$ for $v\in\sphilbfrak{k}_{\sminvtemperature,I}$ is an $\lp^2$ limit of combinations of $\fun{\opfocksegal_t}{f}$ with
$t\in I$, the $\sigma$-algebra $\mblfmlfrak{S}_I$ is generated by $\fun{\opfocksegal}{\sphilbfrak{k}_{\sminvtemperature,I}}$ and the null sets.
For $I=S_{\sminvtemperature}$, Lemma \ref{lem:sharp-time-total} gives $\sphilbfrak{k}_{\sminvtemperature,S_{\sminvtemperature}}=\sphilbfrak{H}_{\sminvtemperature,C}$, and the
continuous linear functionals $\fun{\opfocksegal}{f}$, $f\in\prbqspace_+$, generate the Borel $\sigma$-algebra of the separable space
$\prbqspace_-$.
The identity $E_I=\fun{\Gamma}{e_I}$ is \cite[Proposition 9.28, p.~225]{DerezinskiGerard001}, and the last two identities follow
from \eqref{eq:field-transformations} and the totality of the exponentials $\napiernum^{\imunit\fun{\opfocksegal}{f}}$, $f\in\prbqspace_+$
\cite[Theorem 5.54, p.~125]{DerezinskiGerard001}.
\end{proof}

We use the notions of a generalized path space and of a \(\sminvtemperature\)-Markov path space in \cite[Definitions 21.7 and 21.41, pp.~610 and 624]{DerezinskiGerard001}, where \cite[Definition 21.8, p.~610]{DerezinskiGerard001} defines the \(\sigma\)-algebra of \(I\) as \(\bigvee_{t\in I}U_t\mblfmlfrak{S}_0\); for the path space of Definition \ref{def:singular-gaussian-beta-markov-path-space} it agrees with \(\mblfmlfrak{S}_I\), as shown in the proof of the following theorem.

\begin{thm}[Singular Gaussian $\sminvtemperature$-Markov property]\label{thm:singular-gaussian-beta-markov}
The singular Gaussian $\sminvtemperature$-Markov path space of Definition \ref{def:singular-gaussian-beta-markov-path-space} is a $\sminvtemperature$-Markov path
space, that is,
\begin{equation}\label{eq:beta-markov}
\begin{aligned}
U_{\sminvtemperature}
&=
1,
&
RE_{\partial}
&=
E_{\partial},
&
E_{I_+}E_{I_-}
&=
E_{\partial},
\end{aligned}
\end{equation}
in addition to the axioms of a generalized path space.
\end{thm}

\begin{proof}
By Proposition \ref{prop:symmetries-conditional-expectations}, $U_t$ and $R$ are measure-preserving $\ast$-automorphisms with
$U_{s+t}=U_sU_t$, $U_{\sminvtemperature}=1$, $R^2=1$, and $RU_t=U_{-t}R$, and
$U_t\mblfmlfrak{S}_0=\mblfmlfrak{S}_{\setone{t}}$ by \eqref{eq:field-transformations}; hence $\bigvee_{t\in I}U_t\mblfmlfrak{S}_0=\mblfmlfrak{S}_I$ for
nonempty $I$, and $\mblfmlfrak{S}_{S_{\sminvtemperature}}=\mblfmlfrak{S}$.
The Koopman group is strongly continuous, since the exponentials are total and
$$\begin{aligned}
\norm{\rbk{K_{U_t}-1}\napiernum^{\imunit\fun{\opfocksegal}{v}}}^2
&=
2-2\napiernum^{-\frac{1}{4}\norm{u_tv-v}^2};
\end{aligned}$$
since $U_t$ is isometric on $\lp^\infty$, this also gives the $\sigma$-weak continuity of $t\mapsto U_tF$ for
$F\in\lp^\infty$.
By \eqref{eq:conditional-expectation-second-quantization} and Proposition \ref{prop:half-circle-projections},
$E_{I_+}E_{I_-}=\fun{\Gamma}{e_{I_+}e_{I_-}}=\fun{\Gamma}{e_{\partial}}=E_{\partial}$.
Since $r$ fixes $j_0f$ and $j_{\frac{\sminvtemperature}{2}}f=j_{-\frac{\sminvtemperature}{2}}f$, it is the identity on $\sphilbfrak{k}_{\sminvtemperature,\partial}$, so
$RE_{\partial}=\fun{\Gamma}{re_{\partial}}=\fun{\Gamma}{e_{\partial}}=E_{\partial}$.
\end{proof}

\begin{rem}[Translated laws]\label{rem:translated-laws}
Let $q_0\in\prbqspace_-$ satisfy $\tilde{u}_tq_0=q_0$ for $t\in\fldreal$ and $\tilde{r}q_0=q_0$, and let $\tau_{q_0}\colon q\mapsto q+q_0$.
Then $\pushoutrbk{\tau_{q_0}}\msrprb_{\sminvtemperature,0}$, with the $\sigma$-algebras $\tau_{q_0}\rbk{\mblfmlfrak{S}_I}$ completed for this measure
and the same $U_t$ and $R$, is again a $\sminvtemperature$-Markov path space, because $\tau_{q_0}$ commutes with $\tilde{u}_t$ and $\tilde{r}$
and $F\mapsto F\circ\tau_{q_0}$ intertwines all the structure.
It is mutually singular to $\msrprb_{\sminvtemperature,0}$ if $q_0\notin\sphilbfrak{H}_{\sminvtemperature,C}$, by the Cameron--Martin
theorem \cite[Theorem 2.4.5, p.~61]{VladimirBogachev002}.
\end{rem}

\section{Euclidean realization of the Araki--Woods representations}\label{sec:araki-woods}

In this section, we identify the physical Hilbert space of the singular Gaussian \(\sminvtemperature\)-Markov path space with the Araki--Woods Fock space, intertwining the Weyl operators, the resolvents, the Liouvilleans, and the modular conjugations.

\subsection{Araki--Woods representations and Segal sharp-time fields}\label{arakiwoods-representations-and-segal-sharp-time-fields}

In this subsection, we fix the Araki--Woods objects, the Segal sharp-time fields, and the boundary one-particle map.

Let \(\rbk{\sphilbfrak{h}_{\fldreal}}_{\fldcmp}\) be the complexification of \(\sphilbfrak{h}_{\fldreal}\), that is, the real vector space \(\sphilbfrak{h}_{\fldreal}\times\sphilbfrak{h}_{\fldreal}\), whose element \(\pairbk{f,g}\) is written \(f+\imunit g\), with the multiplication by complex scalars and the inner product given, for \(a,b\in\fldreal\) and \(f,g,f',g'\in\sphilbfrak{h}_{\fldreal}\), by \begin{equation}\label{eq:complexification}
\begin{aligned}
\rbk{a+\imunit b}\rbk{f+\imunit g}
&=
\rbk{af-bg}+\imunit\rbk{bf+ag},
\\ 
\bkt{f+\imunit g}{f'+\imunit g'}
&=
\bkt{f}{f'}_{\sphilbfrak{h}_{\fldreal}}+\bkt{g}{g'}_{\sphilbfrak{h}_{\fldreal}}
+\imunit\rbk{\bkt{f}{g'}_{\sphilbfrak{h}_{\fldreal}}-\bkt{g}{f'}_{\sphilbfrak{h}_{\fldreal}}}.
\end{aligned}
\end{equation} The vectors \(f+\imunit0\), \(f\in\sphilbfrak{h}_{\fldreal}\), are called real vectors and are identified with \(f\). For a real linear operator \(T\) on \(\sphilbfrak{h}_{\fldreal}\), its complex linear extension is \(T\rbk{f+\imunit g}=Tf+\imunit Tg\) on \(\dom T+\imunit\dom T\), and it is denoted by the same symbol; in particular, \(\physham[h]\) and the operators \(\fun{\varphi}{\physham[h]}\) of Definition \ref{def:one-particle-hamiltonian} also denote their complex linear extensions, and \(\fun{\varphi}{\physham[h]}\) is then the functional calculus of the complex linear extension of \(\physham[h]\). The complexification of \(\sphilbfrak{h}_{\fldreal}\), rather than \(\sphilbfrak{h}\) itself, is used because the symplectic form \(\opimag\bkt{f}{g}\) of \(\rbk{\sphilbfrak{h}_{\fldreal}}_{\fldcmp}\) vanishes on the real vectors, as the commutativity of the fields \(\fun{\opfocksegal_0}{\cdot}\) requires. The Planck density is \(\smnumberdensity_{\sminvtemperature}=\rbk{\napiernum^{\sminvtemperature\physham[h]}-1}^{-1}\).

\begin{defn}[Araki--Woods representations]\label{def:araki-woods}
Let $\sphilbfrak{K}_{\txtarakiwoods}=\rbk{\sphilbfrak{h}_{\fldreal}}_{\fldcmp}\oplus\cmpconj{\rbk{\sphilbfrak{h}_{\fldreal}}_{\fldcmp}}$, where, in the second summand,
$\cmpconj{x}$ denotes a vector $x$ regarded as a vector of the conjugate space and $\cmpconj{T}$ denotes an operator $T$ acting on
the conjugate space.
For $f\in\dom\smnumberdensity_{\sminvtemperature}^{\frac{1}{2}}$, let
$$\begin{aligned}
\fun{\ell_{\sminvtemperature}}{f}
&=
\rbk{1+\smnumberdensity_{\sminvtemperature}}^{\frac{1}{2}}f\oplus\cmpconj{\smnumberdensity_{\sminvtemperature}^{\frac{1}{2}}f},
&
\fun{r_{\sminvtemperature}}{f}
&=
\smnumberdensity_{\sminvtemperature}^{\frac{1}{2}}f\oplus\cmpconj{\rbk{1+\smnumberdensity_{\sminvtemperature}}^{\frac{1}{2}}f}.
\end{aligned}$$
On the Araki--Woods Fock space, that is, the bosonic Fock space $\fun{\spfockgamma_{\txtbsn}}{\sphilbfrak{K}_{\txtarakiwoods}}$ of
\cite[Definition 3.27, p.~67]{DerezinskiGerard001}, with vacuum
$\oagnsvector[\Psi]_{\txtarakiwoods}$, let
$\fun{\opfockweyl}{F}=\napiernum^{\imunit\fun{\phi}{F}}$ with the self-adjoint Segal field $\fun{\phi}{F}=2^{-\frac{1}{2}}\rbk{\fun{a^{\ast}}{F}+\fun{a}{F}}$, where
$\fun{a^{\ast}}{F}$ and $\fun{a}{F}$ are the creation and annihilation operators of \cite[Definition 3.50 and Theorem 3.51, p.~75]{DerezinskiGerard001}, so that
$\bkt{\oagnsvector[\Psi]_{\txtarakiwoods}}{\fun{\opfockweyl}{F}\oagnsvector[\Psi]_{\txtarakiwoods}}=\napiernum^{-\frac{1}{4}\norm{F}^2}$.
The left and right Araki--Woods Weyl representations of the Weyl algebras over the real vector space
$\dom\smnumberdensity_{\sminvtemperature}^{\frac{1}{2}}$ with the symplectic forms $\opimag\bkt{\cdot}{\cdot}$ and $-\opimag\bkt{\cdot}{\cdot}$ are given by
$\fun{\oarepn_{\txtarakiwoods,\sminvtemperature,\txtleft}}{\fun{\opfockweyl}{f}}=\fun{\opfockweyl}{\fun{\ell_{\sminvtemperature}}{f}}$ and
$\fun{\oarepn_{\txtarakiwoods,\sminvtemperature,\txtright}}{\fun{\opfockweyl}{f}}=\fun{\opfockweyl}{\fun{r_{\sminvtemperature}}{f}}$, and
$\fun{\phi_{\txtleft,\sminvtemperature}}{f}=\fun{\phi}{\fun{\ell_{\sminvtemperature}}{f}}$ and $\fun{\phi_{\txtright,\sminvtemperature}}{f}=\fun{\phi}{\fun{r_{\sminvtemperature}}{f}}$ are their
Stone generators.
The left and right representations of the resolvent algebras \cite[Definitions 3.1 and 3.4]{BuchholzGrundling002} are
$$\begin{aligned}
\fun{\oarepn_{\txtarakiwoods,\txtleft,\oaresolventalgebra,\sminvtemperature}}{\fun{\oaresolvent}{\lambda,f}}
&=
\rbk{\imunit\lambda-\fun{\phi_{\txtleft,\sminvtemperature}}{f}}^{-1},
\\ 
\fun{\oarepn_{\txtarakiwoods,\txtright,\oaresolventalgebra,\sminvtemperature}}{\fun{\oaresolvent}{\lambda,f}}
&=
\rbk{\imunit\lambda-\fun{\phi_{\txtright,\sminvtemperature}}{f}}^{-1},
\quad\lambda\in\fldreal\setminus\setone{0}.
\end{aligned}$$
With the second quantizations $\fun{\opfocksndqntdiff}{\cdot}$ and $\fun{\opfocksndqnt}{\cdot}$ of \cite[Definitions 3.34 and 3.35, p.~68]{DerezinskiGerard001},
the Liouvillean is $\physliouvilean_{\txtarakiwoods}=\fun{\opfocksndqntdiff}{b}$ with $b=\physham[h]\oplus-\cmpconj{\physham[h]}$, the modular
conjugation is $J_{\txtarakiwoods}=\fun{\opfocksndqnt}{j}$ with $\fun{j}{x\oplus\cmpconj{y}}=y\oplus\cmpconj{x}$, and
$\fun{\oastate[\varphi_{\txtarakiwoods,\sminvtemperature}]}{B}=\bkt{\oagnsvector[\Psi]_{\txtarakiwoods}}{B\oagnsvector[\Psi]_{\txtarakiwoods}}$.
\end{defn}

These are the Araki--Woods representations \cite{ArakiWoods001} in the form of \cite[Section 17.1.5]{DerezinskiGerard001}, and the resolvent-algebra representations are the regular representations corresponding to the regular Weyl representations \cite[Corollary 4.4]{BuchholzGrundling002}.

Let \begin{equation}\label{eq:segal-test-space}
\begin{aligned}
\mathcal{T}_{\fldreal,\physham[h]}
&=
\set{f\in\dom\physham[h]^{\frac{1}{2}}}{\rbk{2\physham[h]}^{\frac{1}{2}}f\in\mathcal{T}_{\fldreal}},
&
\sphilbfrak{h}_{\fldreal,\smnumberdensity}
&=
\sphilbfrak{h}_{\fldreal}\cap\dom\smnumberdensity_{\sminvtemperature}^{\frac{1}{2}},
\end{aligned}
\end{equation} and let \(\norm{f}_{\smnumberdensity}^2=\norm{f}^2+2\norm{\smnumberdensity_{\sminvtemperature}^{\frac{1}{2}}f}^2\) for \(f\in\sphilbfrak{h}_{\fldreal,\smnumberdensity}\). For \(t\in S_{\sminvtemperature}\) and \(f\in\mathcal{T}_{\fldreal,\physham[h]}\), \(\fun{\opfocksegal_t}{\rbk{2\physham[h]}^{\frac{1}{2}}f}\) is the Segal sharp-time field.

\begin{lem}[Segal test space and the domain of the Planck density]\label{lem:segal-test-space}
The following hold.
\begin{enumerate}
\item $\rbk{2\physham[h]}^{\frac{1}{2}}\mathcal{T}_{\fldreal,\physham[h]}=\mathcal{T}_{\fldreal}$ and $\mathcal{T}_{\fldreal,\physham[h]}\subset\sphilbfrak{h}_{\fldreal,\smnumberdensity}$.
\item $\norm{f}_{\smnumberdensity}=\norm{\rbk{\fun{\coth}{\frac{\sminvtemperature\physham[h]}{2}}}^{\frac{1}{2}}f}$, $\sphilbfrak{h}_{\fldreal,\smnumberdensity}$ is complete in this
norm, and $\mathcal{T}_{\fldreal,\physham[h]}$ is dense in it.
\item The map $f\mapsto\rbk{2\physham[h]}^{\frac{1}{2}}f$ from $\mathcal{T}_{\fldreal,\physham[h]}$ into $\sphilbfrak{R}$ extends uniquely to an
isometry $\iota_{\physham[h]}\colon\sphilbfrak{h}_{\fldreal,\smnumberdensity}\to\sphilbfrak{R}$.
\end{enumerate}
In particular,
\begin{equation}\label{eq:segal-equal-time-covariance}
\begin{aligned}
\sqfun{\prbexp_{\msrprb_{\sminvtemperature,0}}}{\fun{\opfocksegal_t}{\rbk{2\physham[h]}^{\frac{1}{2}}f}\fun{\opfocksegal_t}{\rbk{2\physham[h]}^{\frac{1}{2}}g}}
&=
\frac{1}{2}\bkt{f}{\fun{\coth}{\frac{\sminvtemperature\physham[h]}{2}}g}_{\sphilbfrak{h}_{\fldreal}},
\quad f,g\in\mathcal{T}_{\fldreal,\physham[h]}.
\end{aligned}
\end{equation}
\end{lem}

\begin{proof}
All statements are identities of the functional calculus of $\physham[h]$.
Since $\mathcal{T}_{\fldreal}\subset\dom\physham[h]^{-1}\subset\dom\physham[h]^{-\frac{1}{2}}$, the map $\rbk{2\physham[h]}^{-\frac{1}{2}}$ inverts
$\rbk{2\physham[h]}^{\frac{1}{2}}$ on $\mathcal{T}_{\fldreal}$, and $\rbk{\napiernum^{\sminvtemperature\lambda}-1}^{-1}\leq\rbk{\sminvtemperature\lambda}^{-1}$ gives (1).
The identity $1+2\rbk{\napiernum^{\sminvtemperature\lambda}-1}^{-1}=\fun{\coth}{\frac{\sminvtemperature\lambda}{2}}$ gives the norm in (2), and
completeness holds because $\rbk{\fun{\coth}{\frac{\sminvtemperature\physham[h]}{2}}}^{\frac{1}{2}}\geq1$ is closed.
Since $\fun{\varpi}{\lambda}\rbk{2\lambda}^{\frac{1}{2}}=\rbk{\fun{\coth}{\frac{\sminvtemperature\lambda}{2}}}^{\frac{1}{2}}$, one has
$\norm{\rbk{2\physham[h]}^{\frac{1}{2}}f}_{\sphilbfrak{R}}=\norm{f}_{\smnumberdensity}$ for $f\in\mathcal{T}_{\fldreal,\physham[h]}$, and
$\rbk{\fun{\coth}{\frac{\sminvtemperature\physham[h]}{2}}}^{\frac{1}{2}}\mathcal{T}_{\fldreal,\physham[h]}=\fun{\varpi}{\physham[h]}\mathcal{T}_{\fldreal}$ is dense in
$\sphilbfrak{h}_{\fldreal}$ by Lemma \ref{lem:weighted-density}; this gives the density in (2) and the isometric extension in (3).
The covariance \eqref{eq:segal-equal-time-covariance} is \eqref{eq:sharp-time-field-covariance} at $s=t$ with $\fun{M_0}{\lambda}\cdot2\lambda=2\fun{\coth}{\frac{\sminvtemperature\lambda}{2}}$.
\end{proof}

\begin{lem}[Boundary one-particle map]\label{lem:boundary-one-particle-map}
The following hold.
\begin{enumerate}
\item There is a unique continuous real linear map $\iota_{\partial}\colon\sphilbfrak{k}_{\sminvtemperature,\partial}\to\sphilbfrak{K}_{\txtarakiwoods}$ with
\begin{equation}\label{eq:boundary-one-particle-map}
\begin{aligned}
\fun{\iota_{\partial}}{j_0\rbk{2\physham[h]}^{\frac{1}{2}}f+j_{\frac{\sminvtemperature}{2}}\rbk{2\physham[h]}^{\frac{1}{2}}g}
&=
\fun{\ell_{\sminvtemperature}}{f}+\fun{r_{\sminvtemperature}}{g},
\quad f,g\in\mathcal{T}_{\fldreal,\physham[h]}.
\end{aligned}
\end{equation}
It satisfies $\bkt{\fun{\iota_{\partial}}{v}}{\fun{\iota_{\partial}}{w}}_{\sphilbfrak{K}_{\txtarakiwoods}}=\bkt{v}{w}_{\sphilbfrak{H}_{\sminvtemperature,C}}$, and the complex
linear span of its range is dense in $\sphilbfrak{K}_{\txtarakiwoods}$.
\item For $s\in I_+$ and $f\in\mathcal{T}_{\fldreal,\physham[h]}$, $\fun{\ell_{\sminvtemperature}}{f}\in\dom\napiernum^{-sb}$ and
\begin{equation}\label{eq:boundary-semigroup}
\begin{aligned}
\fun{\iota_{\partial}}{e_{\partial}j_s\rbk{2\physham[h]}^{\frac{1}{2}}f}
&=
\napiernum^{-sb}\fun{\ell_{\sminvtemperature}}{f}.
\end{aligned}
\end{equation}
\item The operator $u_{\frac{\sminvtemperature}{2}}r$ maps $\sphilbfrak{k}_{\sminvtemperature,\partial}$ onto itself, and
$\iota_{\partial}\circ u_{\frac{\sminvtemperature}{2}}r=j\circ\iota_{\partial}$ there.
\end{enumerate}
\end{lem}

\begin{proof}
(1): Let $f,f',g,g'\in\mathcal{T}_{\fldreal,\physham[h]}$, and write $A=\rbk{1+\smnumberdensity_{\sminvtemperature}}^{\frac{1}{2}}$ and
$B=\smnumberdensity_{\sminvtemperature}^{\frac{1}{2}}$.
By Proposition \ref{prop:sharp-time-vectors}, $\lambda\fun{M_0}{\lambda}=\fun{\coth}{\frac{\sminvtemperature\lambda}{2}}$, and
$\lambda\fun{M_{\frac{\sminvtemperature}{2}}}{\lambda}=\rbk{\fun{\sinh}{\frac{\sminvtemperature\lambda}{2}}}^{-1}$,
$$\begin{aligned}
\bkt{j_0\rbk{2\physham[h]}^{\frac{1}{2}}f}{j_0\rbk{2\physham[h]}^{\frac{1}{2}}f'}
&=
\bkt{f}{\fun{\coth}{\frac{\sminvtemperature\physham[h]}{2}}f'},
\\ 
\bkt{j_0\rbk{2\physham[h]}^{\frac{1}{2}}f}{j_{\frac{\sminvtemperature}{2}}\rbk{2\physham[h]}^{\frac{1}{2}}g'}
&=
\bkt{f}{\rbk{\fun{\sinh}{\frac{\sminvtemperature\physham[h]}{2}}}^{-1}g'},
\end{aligned}$$
and the same coth-form for $j_{\frac{\sminvtemperature}{2}}$.
For real vectors, $\bkt{\fun{\ell_{\sminvtemperature}}{f}}{\fun{\ell_{\sminvtemperature}}{f'}}=\bkt{f}{\rbk{A^2+B^2}f'}$ and
$\bkt{\fun{\ell_{\sminvtemperature}}{f}}{\fun{r_{\sminvtemperature}}{g'}}=2\bkt{f}{ABg'}$ are real, and
$A^2+B^2=\fun{\coth}{\frac{\sminvtemperature\physham[h]}{2}}$ and $2AB=\rbk{\fun{\sinh}{\frac{\sminvtemperature\physham[h]}{2}}}^{-1}$.
Hence \eqref{eq:boundary-one-particle-map} is isometric on a dense subspace of $\sphilbfrak{k}_{\sminvtemperature,\partial}$ by
Lemma \ref{lem:segal-test-space}(1), and extends uniquely.

For the density, by Lemma \ref{lem:segal-test-space} the closure of the complex span of the range contains
$\fun{\ell_{\sminvtemperature}}{f}+\fun{r_{\sminvtemperature}}{g}$ for $f,g\in\sphilbfrak{h}_{\fldreal,\smnumberdensity}$, because
$\norm{\fun{\ell_{\sminvtemperature}}{f}}=\norm{\fun{r_{\sminvtemperature}}{f}}=\norm{f}_{\smnumberdensity}$.
Let $0<\varepsilon<M$, and let $x=x_1+\imunit x_2$ and $y=y_1+\imunit y_2$ with $x_i,y_i$ in the range of the
spectral projection $\fun{\msrprb_{\physham[h]}}{\closedinterval{\varepsilon}{M}}$ on $\sphilbfrak{h}_{\fldreal}$, which is contained in
$\sphilbfrak{h}_{\fldreal,\smnumberdensity}$.
Since $A$ and $B$ are bounded on this range with $A^2-B^2=1$,
$$\begin{aligned}
x\oplus\cmpconj{y}
&=
\fun{\ell_{\sminvtemperature}}{Ax_1-By_1}+\imunit\fun{\ell_{\sminvtemperature}}{Ax_2+By_2}
+\fun{r_{\sminvtemperature}}{Ay_1-Bx_1}-\imunit\fun{r_{\sminvtemperature}}{Ay_2+Bx_2},
\end{aligned}$$
where the scalar multiplication of the conjugate space has been used in the second summand.
These vectors $x\oplus\cmpconj{y}$ are dense in $\sphilbfrak{K}_{\txtarakiwoods}$.

(2): Since $\napiernum^{2s\lambda}\rbk{\napiernum^{\sminvtemperature\lambda}-1}^{-1}\leq1+\rbk{\sminvtemperature\lambda}^{-1}$ for $s\leq\frac{\sminvtemperature}{2}$ and
$\mathcal{T}_{\fldreal,\physham[h]}\subset\dom\physham[h]^{-\frac{1}{2}}$, the vector $Bf$ lies in $\dom\napiernum^{s\physham[h]}$, so
$\fun{\ell_{\sminvtemperature}}{f}\in\dom\napiernum^{-sb}$, and $\napiernum^{-sb}\fun{\ell_{\sminvtemperature}}{f}=A\napiernum^{-s\physham[h]}f\oplus\cmpconj{B\napiernum^{s\physham[h]}f}$.
Let $u=\rbk{2\physham[h]}^{\frac{1}{2}}f$.
Taking adjoints in \eqref{eq:half-circle-projection-product} gives $e_{\partial}=e_{I_-}e_{I_+}$, so $e_{\partial}j_su=e_{I_-}j_su=re_{I_+}j_{-s}u$, and
\eqref{eq:half-circle-projection} gives
$e_{\partial}j_su=j_0\fun{\eta_{0,-s}}{\physham[h]}u+j_{\frac{\sminvtemperature}{2}}\fun{\eta_{1,-s}}{\physham[h]}u$.
By \eqref{eq:boundary-one-particle-map}, its image is $\rbk{A\eta_{0,-s}+B\eta_{1,-s}}f\oplus\cmpconj{\rbk{B\eta_{0,-s}+A\eta_{1,-s}}f}$ with the
functions evaluated at $\physham[h]$.
The exponential form of $\sinh$ gives
$A\eta_{0,-s}+B\eta_{1,-s}=A\napiernum^{-s\lambda}$ and $B\eta_{0,-s}+A\eta_{1,-s}=B\napiernum^{s\lambda}$ for $\lambda>0$, which is
\eqref{eq:boundary-semigroup}.

(3): By Proposition \ref{prop:sharp-time-vectors}, $u_{\frac{\sminvtemperature}{2}}r$ exchanges $j_0u$ and $j_{\frac{\sminvtemperature}{2}}u$, and it is an
involution.
Since $\fun{j}{\fun{\ell_{\sminvtemperature}}{f}}=\fun{r_{\sminvtemperature}}{f}$ and $\fun{j}{\fun{r_{\sminvtemperature}}{f}}=\fun{\ell_{\sminvtemperature}}{f}$ for real $f$, the two
continuous real linear maps agree on the generators.
\end{proof}

\subsection{The Euclidean realization theorem}\label{the-euclidean-realization-theorem}

In this subsection, we prove the unitary equivalence and the imaginary-time correlation identities.

Let \(\sphilbfrak{H}_{\sminvtemperature,0}=\fun{\lp^2}{\prbqspace_-,\mblfmlfrak{S}_{\partial},\msrprb_{\sminvtemperature,0}}\) be the physical Hilbert space \cite[Definition 21.42, p.~624]{DerezinskiGerard001}, and let \(\Psi_{\sminvtemperature,0}=1\). By Theorem \ref{thm:singular-gaussian-beta-markov} and \cite[Theorem 21.43 and Definition 21.44, pp.~625--626]{DerezinskiGerard001}, there is a unique self-adjoint operator \(\physliouvilean_{\txteuclid}\) on \(\sphilbfrak{H}_{\sminvtemperature,0}\) such that, for \(0\leq t<\frac{\sminvtemperature}{2}\) and \(F\in\fun{\lp^2}{\prbqspace_-,\mblfmlfrak{S}_{\closedinterval{0}{\frac{\sminvtemperature}{2}-t}},\msrprb_{\sminvtemperature,0}}\), \begin{equation}\label{eq:euclidean-liouvillean}
\begin{aligned}
E_{\partial}F
&\in
\dom\napiernum^{-t\physliouvilean_{\txteuclid}},
&
\napiernum^{-t\physliouvilean_{\txteuclid}}E_{\partial}F
&=
E_{\partial}U_tF.
\end{aligned}
\end{equation} Let \(J_{\txteuclid}\) be the antiunitary operator on \(\sphilbfrak{H}_{\sminvtemperature,0}\) defined by \[\begin{aligned}
J_{\txteuclid}E_{\partial}F
&=
E_{\partial}U_{\frac{\sminvtemperature}{2}}R\cmpconj{F},
\quad F\in\fun{\lp^2}{\prbqspace_-,\mblfmlfrak{S},\msrprb_{\sminvtemperature,0}};
\end{aligned}\] cf.~\cite[(21.29), (21.30), and Definition 21.46, p.~626]{DerezinskiGerard001}. For \(t\in\partial\), \(f\in\mathcal{T}_{\fldreal,\physham[h]}\), and \(\lambda\in\fldreal\setminus\setone{0}\), the Euclidean resolvent is the bounded \(\mblfmlfrak{S}_{\setone{t}}\)-measurable function \(\fun{\oaresolvent_{\txteuclid,t}}{\lambda,f}=\rbk{\imunit\lambda-\fun{\opfocksegal_t}{\rbk{2\physham[h]}^{\frac{1}{2}}f}}^{-1}\), which acts on \(\sphilbfrak{H}_{\sminvtemperature,0}\) by multiplication.

\begin{thm}[Euclidean realization of the Araki--Woods representations]\label{thm:araki-woods-equivalence}
There is a unique unitary $T_{\txteuclid}\colon\sphilbfrak{H}_{\sminvtemperature,0}\to\fun{\spfockgamma_{\txtbsn}}{\sphilbfrak{K}_{\txtarakiwoods}}$ such that
\begin{equation}\label{eq:weyl-intertwining}
\begin{aligned}
T_{\txteuclid}\Psi_{\sminvtemperature,0}
&=
\oagnsvector[\Psi]_{\txtarakiwoods},
\\ 
T_{\txteuclid}\napiernum^{\imunit\fun{\opfocksegal_0}{\rbk{2\physham[h]}^{\frac{1}{2}}f}}T_{\txteuclid}^{-1}
&=
\fun{\oarepn_{\txtarakiwoods,\sminvtemperature,\txtleft}}{\fun{\opfockweyl}{f}},
\\ 
T_{\txteuclid}\napiernum^{\imunit\fun{\opfocksegal_{\frac{\sminvtemperature}{2}}}{\rbk{2\physham[h]}^{\frac{1}{2}}g}}T_{\txteuclid}^{-1}
&=
\fun{\oarepn_{\txtarakiwoods,\sminvtemperature,\txtright}}{\fun{\opfockweyl}{g}},
\quad f,g\in\mathcal{T}_{\fldreal,\physham[h]}.
\end{aligned}
\end{equation}
For $f,g\in\mathcal{T}_{\fldreal,\physham[h]}$ and $\lambda,\mu\in\fldreal\setminus\setone{0}$, it satisfies
\begin{equation}\label{eq:resolvent-intertwining}
\begin{aligned}
T_{\txteuclid}\fun{\oaresolvent_{\txteuclid,0}}{\lambda,f}T_{\txteuclid}^{-1}
&=
\fun{\oarepn_{\txtarakiwoods,\txtleft,\oaresolventalgebra,\sminvtemperature}}{\fun{\oaresolvent}{\lambda,f}},
&
T_{\txteuclid}\fun{\oaresolvent_{\txteuclid,\frac{\sminvtemperature}{2}}}{\mu,g}T_{\txteuclid}^{-1}
&=
\fun{\oarepn_{\txtarakiwoods,\txtright,\oaresolventalgebra,\sminvtemperature}}{\fun{\oaresolvent}{\mu,g}},
\end{aligned}
\end{equation}
and
\begin{equation}\label{eq:modular-intertwining}
\begin{aligned}
T_{\txteuclid}\physliouvilean_{\txteuclid}
&=
\physliouvilean_{\txtarakiwoods}T_{\txteuclid},
&
T_{\txteuclid}J_{\txteuclid}
&=
J_{\txtarakiwoods}T_{\txteuclid}.
\end{aligned}
\end{equation}
In particular,
\begin{equation}\label{eq:boundary-expectations}
\begin{aligned}
\fun{\oastate[\varphi_{\txtarakiwoods,\sminvtemperature}]}{
\fun{\oarepn_{\txtarakiwoods,\sminvtemperature,\txtleft}}{\fun{\opfockweyl}{f}}
\fun{\oarepn_{\txtarakiwoods,\sminvtemperature,\txtright}}{\fun{\opfockweyl}{g}}}
&=
\sqfun{\prbexp_{\msrprb_{\sminvtemperature,0}}}{
\napiernum^{\imunit\fun{\opfocksegal_0}{\rbk{2\physham[h]}^{\frac{1}{2}}f}}
\napiernum^{\imunit\fun{\opfocksegal_{\frac{\sminvtemperature}{2}}}{\rbk{2\physham[h]}^{\frac{1}{2}}g}}},
\end{aligned}
\end{equation}
and the same identity holds with the resolvents of \eqref{eq:resolvent-intertwining} in place of the Weyl operators.
\end{thm}

\begin{proof}
\emph{Construction.}
For $v,w\in\sphilbfrak{k}_{\sminvtemperature,\partial}$, the function $\napiernum^{\imunit\fun{\opfocksegal}{v}}$ belongs to $\sphilbfrak{H}_{\sminvtemperature,0}$ by
Proposition \ref{prop:symmetries-conditional-expectations}, and \eqref{eq:characteristic-functional} gives
$\bkt{\napiernum^{\imunit\fun{\opfocksegal}{v}}}{\napiernum^{\imunit\fun{\opfocksegal}{w}}}=\napiernum^{-\frac{1}{4}\norm{w-v}^2}$.
By Lemma \ref{lem:boundary-one-particle-map}(1), the inner products of the vectors $\fun{\iota_{\partial}}{v}$ are real, so the Weyl relations
give the same inner products for the vectors $\fun{\opfockweyl}{\fun{\iota_{\partial}}{v}}\oagnsvector[\Psi]_{\txtarakiwoods}$.
The functions $\napiernum^{\imunit\fun{\opfocksegal}{v}}$, $v\in\sphilbfrak{k}_{\sminvtemperature,\partial}$, are total in $\sphilbfrak{H}_{\sminvtemperature,0}$, because
the exponentials $\napiernum^{\imunit\fun{\opfocksegal}{f}}$, $f\in\prbqspace_+$, are total in $\fun{\lp^2}{\prbqspace_-,\msrprb_{\sminvtemperature,0}}$
\cite[Theorem 5.54, p.~125]{DerezinskiGerard001} and $E_{\partial}=\fun{\Gamma}{e_{\partial}}$ maps them to multiples of
$\napiernum^{\imunit\fun{\opfocksegal}{e_{\partial}f}}$.
The coherent vectors $\fun{\opfockweyl}{G}\oagnsvector[\Psi]_{\txtarakiwoods}$ with $G$ in a real subspace whose complex span is dense are total in the
Araki--Woods Fock space $\fun{\spfockgamma_{\txtbsn}}{\sphilbfrak{K}_{\txtarakiwoods}}$, by differentiating $t\mapsto\fun{\opfockweyl}{tG}\oagnsvector[\Psi]_{\txtarakiwoods}$ at $t=0$ and polarizing.
Hence there is a unique unitary with
\begin{equation}\label{eq:euclidean-exponential-map}
\begin{aligned}
T_{\txteuclid}\napiernum^{\imunit\fun{\opfocksegal}{v}}
&=
\fun{\opfockweyl}{\fun{\iota_{\partial}}{v}}\oagnsvector[\Psi]_{\txtarakiwoods},
\quad v\in\sphilbfrak{k}_{\sminvtemperature,\partial}.
\end{aligned}
\end{equation}

\emph{Weyl operators and resolvents.}
For $f\in\mathcal{T}_{\fldreal,\physham[h]}$, multiplication by $\napiernum^{\imunit\fun{\opfocksegal_0}{\rbk{2\physham[h]}^{\frac{1}{2}}f}}$ maps
$\napiernum^{\imunit\fun{\opfocksegal}{v}}$ to $\napiernum^{\imunit\fun{\opfocksegal}{v+j_0\rbk{2\physham[h]}^{\frac{1}{2}}f}}$, and
$\fun{\opfockweyl}{\fun{\ell_{\sminvtemperature}}{f}}\fun{\opfockweyl}{\fun{\iota_{\partial}}{v}}=\fun{\opfockweyl}{\fun{\ell_{\sminvtemperature}}{f}+\fun{\iota_{\partial}}{v}}$, since the
imaginary part of the inner product vanishes; this gives \eqref{eq:weyl-intertwining}, and the right relation is the same
with $r_{\sminvtemperature}$.
The vectors $\napiernum^{\imunit\fun{\opfocksegal_0}{u}}\napiernum^{\imunit\fun{\opfocksegal_{\frac{\sminvtemperature}{2}}}{u'}}$, $u,u'\in\mathcal{T}_{\fldreal}$, are the functions
$\napiernum^{\imunit\fun{\opfocksegal}{j_0u+j_{\frac{\sminvtemperature}{2}}u'}}$, which are total, so \eqref{eq:weyl-intertwining} determines $T_{\txteuclid}$.
Applying \eqref{eq:weyl-intertwining} to $tf$, $t\in\fldreal$, shows that $T_{\txteuclid}$ intertwines the self-adjoint generators
$\fun{\opfocksegal_0}{\rbk{2\physham[h]}^{\frac{1}{2}}f}$ and $\fun{\phi_{\txtleft,\sminvtemperature}}{f}$, and the bounded functional calculus gives
\eqref{eq:resolvent-intertwining}; the expectation identities are the vacuum expectations of these relations.

\emph{Liouvillean.}
Fix $0<t<\frac{\sminvtemperature}{2}$, let $w$ be a finite real linear combination of vectors $j_\sigma u$ with
$\sigma\in\closedinterval{0}{\frac{\sminvtemperature}{2}-t}$ and $u\in\mathcal{T}_{\fldreal}$, and let $F=\napiernum^{\imunit\fun{\opfocksegal}{w}}$ and
$0\leq s\leq t$.
By \eqref{eq:euclidean-liouvillean}, \eqref{eq:gaussian-second-quantization}, and \eqref{eq:field-transformations},
$$\begin{aligned}
\napiernum^{-s\physliouvilean_{\txteuclid}}E_{\partial}F
&=
E_{\partial}U_sF
=
\napiernum^{-\frac{1}{4}\norm{\rbk{1-e_{\partial}}u_sw}^2}\napiernum^{\imunit\fun{\opfocksegal}{e_{\partial}u_sw}}.
\end{aligned}$$
Lemma \ref{lem:boundary-one-particle-map}(2) gives $G=\fun{\iota_{\partial}}{e_{\partial}w}\in\dom\napiernum^{-sb}$ and
$\napiernum^{-sb}G=\fun{\iota_{\partial}}{e_{\partial}u_sw}$.
On the Araki--Woods Fock space, the coherent-vector expansion of $\fun{\opfockweyl}{G}\oagnsvector[\Psi]_{\txtarakiwoods}$ gives
$$\begin{aligned}
\napiernum^{-s\physliouvilean_{\txtarakiwoods}}\fun{\opfockweyl}{G}\oagnsvector[\Psi]_{\txtarakiwoods}
&=
\napiernum^{\frac{1}{4}\rbk{\norm{\napiernum^{-sb}G}^2-\norm{G}^2}}\fun{\opfockweyl}{\napiernum^{-sb}G}\oagnsvector[\Psi]_{\txtarakiwoods},
\end{aligned}$$
and $\norm{u_sw}=\norm{w}$, so
$T_{\txteuclid}\napiernum^{-s\physliouvilean_{\txteuclid}}E_{\partial}F=\napiernum^{-s\physliouvilean_{\txtarakiwoods}}T_{\txteuclid}E_{\partial}F$ for $0\leq s\leq t$.
Let $y=T_{\txteuclid}E_{\partial}F$ and $\physliouvilean_1=T_{\txteuclid}\physliouvilean_{\txteuclid}T_{\txteuclid}^{-1}$.
The functions $z\mapsto\napiernum^{-z\physliouvilean_1}y$ and $z\mapsto\napiernum^{-z\physliouvilean_{\txtarakiwoods}}y$ are continuous on
$\set{z\in\fldcmp}{0\leq\opreal z\leq t}$, holomorphic in its interior, and agree on $\closedinterval{0}{t}$; hence
$\napiernum^{\imunit\theta\physliouvilean_1}y=\napiernum^{\imunit\theta\physliouvilean_{\txtarakiwoods}}y$ for $\theta\in\fldreal$.
These vectors $E_{\partial}F$, for all $0<t<\frac{\sminvtemperature}{2}$, are total in $\sphilbfrak{H}_{\sminvtemperature,0}$, because
$\sigma\mapsto j_\sigma u$ is continuous and $E_{\partial}$ maps the exponentials of $\fun{\opfocksegal}{\sphilbfrak{k}_{\sminvtemperature,I_+}}$ onto multiples of
those of $\fun{\opfocksegal}{\sphilbfrak{k}_{\sminvtemperature,\partial}}$.
Hence $\physliouvilean_1=\physliouvilean_{\txtarakiwoods}$.

\emph{Modular conjugation.}
For $v\in\sphilbfrak{k}_{\sminvtemperature,\partial}$, \eqref{eq:field-transformations} and Lemma \ref{lem:boundary-one-particle-map}(3) give
$J_{\txteuclid}\napiernum^{\imunit\fun{\opfocksegal}{v}}=\napiernum^{-\imunit\fun{\opfocksegal}{u_{\frac{\sminvtemperature}{2}}rv}}$, while
$J_{\txtarakiwoods}\fun{\opfockweyl}{G}\oagnsvector[\Psi]_{\txtarakiwoods}=\fun{\opfockweyl}{-jG}\oagnsvector[\Psi]_{\txtarakiwoods}$ by the antilinearity of $j^{\otimes n}$.
Lemma \ref{lem:boundary-one-particle-map}(3) gives $T_{\txteuclid}J_{\txteuclid}=J_{\txtarakiwoods}T_{\txteuclid}$ on a total set, hence everywhere.
\end{proof}

\begin{rem}[Comparison with the Euclidean construction]\label{rem:comparison-euclidean-construction}
Theorem \ref{thm:araki-woods-equivalence} corresponds to \cite[Theorem 21.58, p.~635]{DerezinskiGerard001}, with $\sphilbfrak{h}_{\fldreal}$ and $\physham[h]$
in place of the real Hilbert space and the one-particle operator there.
The time-zero fields there are indexed by $\sphilbfrak{R}$, as noted after Definition \ref{def:sharp-time-vectors}, and the Fock space there is over
$\rbk{2\physham[h]}^{\frac{1}{2}}$ applied to the complexification of $\sphilbfrak{h}_{\fldreal}$, plus its conjugate, in the notation of
\cite[Definition 2.59, p.~47]{DerezinskiGerard001}.
Under the unitary map induced by $\rbk{2\physham[h]}^{-\frac{1}{2}}$ \cite[Proposition 2.60(2), p.~47]{DerezinskiGerard001}, this Fock space becomes the Araki--Woods
Fock space $\fun{\spfockgamma_{\txtbsn}}{\sphilbfrak{K}_{\txtarakiwoods}}$, and the vector $g$ there corresponds to $f=\rbk{2\physham[h]}^{-\frac{1}{2}}g$, which is the parametrization
$f\mapsto\rbk{2\physham[h]}^{\frac{1}{2}}f$ of the Segal sharp-time field and of its extension $\iota_{\physham[h]}$.
In the formula for the left Araki--Woods Weyl operators there, the first component is printed as $\rbk{1+2\rho}^{\frac{1}{2}}g$, whereas
\cite[Definition 17.22, p.~437]{DerezinskiGerard001} and $\ell_{\sminvtemperature}$ have the factor $\rbk{1+\rho}^{\frac{1}{2}}$.
The latter is the one compatible with \eqref{eq:segal-equal-time-covariance}, since
$\norm{\rbk{1+2\smnumberdensity_{\sminvtemperature}}^{\frac{1}{2}}f\oplus\cmpconj{\smnumberdensity_{\sminvtemperature}^{\frac{1}{2}}f}}^2
=\norm{f}^2+3\norm{\smnumberdensity_{\sminvtemperature}^{\frac{1}{2}}f}^2$ differs from $\norm{\fun{\ell_{\sminvtemperature}}{f}}^2=\norm{f}_{\smnumberdensity}^2$ for
$0\neq f\in\sphilbfrak{h}_{\fldreal,\smnumberdensity}$, because $\smnumberdensity_{\sminvtemperature}^{\frac{1}{2}}$ is injective; the operator factor
$1+2\rho$ in the exponent of the proof of that theorem \cite[p.~636]{DerezinskiGerard001} also requires the factor $\rbk{1+\rho}^{\frac{1}{2}}$.
The generalized Gaussian measure there has covariance $\rbk{D_t^2+\physham[h]^2}^{-1}$, whereas ours has the Segal normalization.
The identifications of the time-$\frac{\sminvtemperature}{2}$ fields and of the resolvents, and Proposition \ref{prop:imaginary-time-correlations},
are not stated there.
\end{rem}

For \(n\in\semigrposint\) and \(0\leq t_1\leq\cdots\leq t_n\leq\frac{\sminvtemperature}{2}\), set \(t_0=t_1\), and for bounded operators \(Y_1,\ldots,Y_n\) write the ordered product \(\prod_{j=1}^{n}\rbk{\napiernum^{-\rbk{t_j-t_{j-1}}\physliouvilean_{\txtarakiwoods}}Y_j}\oagnsvector[\Psi]_{\txtarakiwoods}\) for the vector \(x_1\) defined by \(x_n=Y_n\oagnsvector[\Psi]_{\txtarakiwoods}\) and \(x_k=Y_k\napiernum^{-\rbk{t_{k+1}-t_k}\physliouvilean_{\txtarakiwoods}}x_{k+1}\), \(k=n-1,\ldots,1\), whenever each \(x_{k+1}\) lies in the domain of the exponential applied to it.

\begin{prop}[Imaginary-time correlation identities]\label{prop:imaginary-time-correlations}
Let $n\in\semigrposint$, $0\leq t_1\leq\cdots\leq t_n\leq\frac{\sminvtemperature}{2}$, $f_1,\ldots,f_n\in\mathcal{T}_{\fldreal,\physham[h]}$, and
$\lambda_1,\ldots,\lambda_n\in\fldreal\setminus\setone{0}$, and let either $Y_j=\fun{\oarepn_{\txtarakiwoods,\sminvtemperature,\txtleft}}{\fun{\opfockweyl}{f_j}}$ for
every $j$ or $Y_j=\fun{\oarepn_{\txtarakiwoods,\txtleft,\oaresolventalgebra,\sminvtemperature}}{\fun{\oaresolvent}{\lambda_j,f_j}}$ for every $j$.
Then the ordered product is defined, and
\begin{align}
&\bkt{\oagnsvector[\Psi]_{\txtarakiwoods}}{
\prod_{j=1}^{n}\rbk{\napiernum^{-\rbk{t_j-t_{j-1}}\physliouvilean_{\txtarakiwoods}}
\fun{\oarepn_{\txtarakiwoods,\sminvtemperature,\txtleft}}{\fun{\opfockweyl}{f_j}}}
\oagnsvector[\Psi]_{\txtarakiwoods}}
\notag
\\ 
&=
\sqfun{\prbexp_{\msrprb_{\sminvtemperature,0}}}{
\prod_{j=1}^{n}\napiernum^{\imunit\fun{\opfocksegal_{t_j}}{\rbk{2\physham[h]}^{\frac{1}{2}}f_j}}},
\label{eq:weyl-correlation}
\\ 
&\bkt{\oagnsvector[\Psi]_{\txtarakiwoods}}{
\prod_{j=1}^{n}\rbk{\napiernum^{-\rbk{t_j-t_{j-1}}\physliouvilean_{\txtarakiwoods}}
\fun{\oarepn_{\txtarakiwoods,\txtleft,\oaresolventalgebra,\sminvtemperature}}{\fun{\oaresolvent}{\lambda_j,f_j}}}
\oagnsvector[\Psi]_{\txtarakiwoods}}
\notag
\\ 
&=
\sqfun{\prbexp_{\msrprb_{\sminvtemperature,0}}}{
\prod_{j=1}^{n}\rbk{\imunit\lambda_j-\fun{\opfocksegal_{t_j}}{\rbk{2\physham[h]}^{\frac{1}{2}}f_j}}^{-1}}.
\label{eq:resolvent-correlation}
\end{align}
\end{prop}

\begin{proof}
Let $A_j$ be $\napiernum^{\imunit\fun{\opfocksegal_0}{\rbk{2\physham[h]}^{\frac{1}{2}}f_j}}$ or $\fun{\oaresolvent_{\txteuclid,0}}{\lambda_j,f_j}$, so that
$U_tA_j$ is the corresponding function at time $t$ by \eqref{eq:field-transformations}.
For $k=n,\ldots,1$, let $G_k=\prod_{j=k}^{n}U_{t_j-t_k}A_j$, which is bounded and $\mblfmlfrak{S}_{\closedinterval{0}{t_n-t_k}}$-measurable.
Downward induction shows that the vectors $v_n=A_n$ and $v_k=A_k\napiernum^{-\rbk{t_{k+1}-t_k}\physliouvilean_{\txteuclid}}v_{k+1}$ are defined and
$v_k=E_{\partial}G_k$.
Indeed, if $\delta=t_{k+1}-t_k<\frac{\sminvtemperature}{2}$, then \eqref{eq:euclidean-liouvillean} applies to $G_{k+1}$ because
$\delta+\rbk{t_n-t_{k+1}}\leq\frac{\sminvtemperature}{2}$; if $\delta=\frac{\sminvtemperature}{2}$, then $G_{k+1}$ is $\mblfmlfrak{S}_0$-measurable,
\eqref{eq:euclidean-liouvillean} holds for every $\delta'<\frac{\sminvtemperature}{2}$, and the limit $\delta'\uparrow\frac{\sminvtemperature}{2}$ with respect to the
spectral measure $\msrprb_{\physliouvilean_{\txteuclid}}$ gives the identity at $\delta$.
Since $A_k$ is bounded and $\mblfmlfrak{S}_{\partial}$-measurable, $A_kE_{\partial}U_{\delta}G_{k+1}=E_{\partial}G_k$.
Then $\bkt{1}{v_1}=\sqfun{\prbexp_{\msrprb_{\sminvtemperature,0}}}{G_1}=\sqfun{\prbexp_{\msrprb_{\sminvtemperature,0}}}{\prod_jU_{t_j}A_j}$ by the invariance of
$\msrprb_{\sminvtemperature,0}$ under $U_{t_1}$.
Theorem \ref{thm:araki-woods-equivalence} transfers the vectors $v_k$, the operators $A_k$, and $\napiernum^{-s\physliouvilean_{\txteuclid}}$ with their
domains to the Araki--Woods Fock space.
\end{proof}

\begin{cor}[Extension to the domain of the Planck density]\label{cor:planck-domain-extension}
For $f\in\sphilbfrak{h}_{\fldreal,\smnumberdensity}$ and $t\in S_{\sminvtemperature}$, the field $\fun{\opfocksegal_t}{\fun{\iota_{\physham[h]}}{f}}$ is
$\mblfmlfrak{S}_{\setone{t}}$-measurable and $U_\tau\fun{\opfocksegal_t}{\fun{\iota_{\physham[h]}}{f}}=\fun{\opfocksegal_{t+\tau}}{\fun{\iota_{\physham[h]}}{f}}$.
Theorem \ref{thm:araki-woods-equivalence} and Proposition \ref{prop:imaginary-time-correlations} hold for $f,g,f_1,\ldots,f_n\in\sphilbfrak{h}_{\fldreal,\smnumberdensity}$,
with $\fun{\iota_{\physham[h]}}{f}$ in place of $\rbk{2\physham[h]}^{\frac{1}{2}}f$.
\end{cor}

\begin{proof}
Let $f_m\in\mathcal{T}_{\fldreal,\physham[h]}$ converge to $f$ in $\norm{\cdot}_{\smnumberdensity}$ (Lemma \ref{lem:segal-test-space}).
Since $\norm{\fun{\opfocksegal_t}{\fun{\iota_{\physham[h]}}{f_m}}-\fun{\opfocksegal_t}{\fun{\iota_{\physham[h]}}{f}}}_{\lp^2}=2^{-\frac{1}{2}}\norm{f_m-f}_{\smnumberdensity}$,
the measurability and the translation relation follow from those for $f_m$, and multiplication by the exponentials converges strongly
along almost surely convergent subsequences.
On the Araki--Woods Fock space, $\norm{\fun{\ell_{\sminvtemperature}}{f_m}-\fun{\ell_{\sminvtemperature}}{f}}=\norm{f_m-f}_{\smnumberdensity}$, and the Weyl operators depend
strongly continuously on their arguments; this extends \eqref{eq:weyl-intertwining}.
The resolvent identities follow as in the proof of Theorem \ref{thm:araki-woods-equivalence}, and the proof of
Proposition \ref{prop:imaginary-time-correlations} applies unchanged.
\end{proof}

\section{Quasi-local construction}\label{sec:quasi-local}

In this section, we show that the path spaces of Definition \ref{def:singular-gaussian-beta-markov-path-space} for a sequence of one-particle Hamiltonians joined by isometric embeddings form a projective system, as for the nested momentum lattices of finite-volume approximations, and we construct its projective limit.

\begin{defn}[Directed system of one-particle Hamiltonians]\label{def:directed-system}
For $n\in\monnat$, let $\sphilbfrak{h}_n$ and $\physham[h]_n$ be as in Definition \ref{def:one-particle-hamiltonian}, and let
$\mathcal{T}_{\fldreal,n}$ be a test space for $\physham[h]_n$ in the sense of Definition \ref{def:test-space}.
For $n\leq m$, let $\tilde{\iota}_{n,m}\colon\sphilbfrak{h}_n\to\sphilbfrak{h}_m$ be complex-linear isometries such that, for
$n\leq m\leq l$ and Borel sets $E\subset\fldreal$,
\begin{equation}\label{eq:directed-system}
\begin{aligned}
\tilde{\iota}_{n,n}
&=
1,
\\ 
\tilde{\iota}_{m,l}\tilde{\iota}_{n,m}
&=
\tilde{\iota}_{n,l},
\\ 
\tilde{\iota}_{n,m}\fun{\msrprb_{\physham[h]_n}}{E}
&=
\fun{\msrprb_{\physham[h]_m}}{E}\tilde{\iota}_{n,m},
\\ 
\tilde{\iota}_{n,m}\mathcal{T}_{\fldreal,n}
&\subset
\mathcal{T}_{\fldreal,m}.
\end{aligned}
\end{equation}
The objects of Sections \ref{sec:setting} and \ref{sec:sharp-time} for $\physham[h]_n$ carry the subscript $n$: $C_n$,
$\sphilbfrak{H}_{\sminvtemperature,C_n}$, $V_n$, and $j_{n,t}$; the operators $u_\tau$ and $r$ keep their names on
each space.
For $n\leq m$, let $\iota_{\sminvtemperature,n,m}$ be the isometry
$\fun{\rbk{\iota_{\sminvtemperature,n,m}F}}{t}=\tilde{\iota}_{n,m}\fun{F}{t}$ from
$\fun{\lp^2}{S_{\sminvtemperature};\sphilbfrak{h}_{n,\fldreal}}$ into $\fun{\lp^2}{S_{\sminvtemperature};\sphilbfrak{h}_{m,\fldreal}}$.
\end{defn}

\begin{lem}[Embeddings of the covariance Hilbert spaces]\label{lem:covariance-embeddings}
Let $n\leq m\leq l$.
Then $\iota_{\sminvtemperature,n,m}$ maps $\dom C_n$ into $\dom C_m$, its restriction to $\dom C_n$ extends to an isometry
$\overline{\iota}_{\sminvtemperature,n,m}\colon\sphilbfrak{H}_{\sminvtemperature,C_n}\to\sphilbfrak{H}_{\sminvtemperature,C_m}$, and, for
$\tau\in\fldreal$, $t\in S_{\sminvtemperature}$, and $g\in\dom\physham[h]_n^{-1}$,
\begin{equation}\label{eq:covariance-embeddings}
\begin{aligned}
V_m\overline{\iota}_{\sminvtemperature,n,m}
&=
\iota_{\sminvtemperature,n,m}V_n,
\\ 
\overline{\iota}_{\sminvtemperature,n,m}j_{n,t}g
&=
j_{m,t}\tilde{\iota}_{n,m}g,
\\ 
\overline{\iota}_{\sminvtemperature,n,m}u_\tau
&=
u_\tau\overline{\iota}_{\sminvtemperature,n,m},
\quad
\overline{\iota}_{\sminvtemperature,n,m}r
=
r\overline{\iota}_{\sminvtemperature,n,m},
\\ 
\overline{\iota}_{\sminvtemperature,n,n}
&=
1,
\quad
\overline{\iota}_{\sminvtemperature,m,l}\overline{\iota}_{\sminvtemperature,n,m}
=
\overline{\iota}_{\sminvtemperature,n,l}.
\end{aligned}
\end{equation}
\end{lem}

\begin{proof}
Comparing the matrix elements of the functional calculus, the third line of \eqref{eq:directed-system} gives
$\tilde{\iota}_{n,m}\dom\fun{\varphi}{\physham[h]_n}\subset\dom\fun{\varphi}{\physham[h]_m}$ and
$\fun{\varphi}{\physham[h]_m}\tilde{\iota}_{n,m}=\tilde{\iota}_{n,m}\fun{\varphi}{\physham[h]_n}$ there, for every real Borel function
$\varphi$ on $\openinterval{0}{\infty}$.
Since $\rbk{\iota_{\sminvtemperature,n,m}F}_k=\tilde{\iota}_{n,m}F_k$, $k\in\ringratint$, Lemma \ref{lem:covariance-fourier} gives
$\iota_{\sminvtemperature,n,m}\dom C_n^{\frac{1}{2}}\subset\dom C_m^{\frac{1}{2}}$ and
$C_m^{\frac{1}{2}}\iota_{\sminvtemperature,n,m}=\iota_{\sminvtemperature,n,m}C_n^{\frac{1}{2}}$ on $\dom C_n^{\frac{1}{2}}$, and hence the
same for $C=C^{\frac{1}{2}}C^{\frac{1}{2}}$.
Thus $\iota_{\sminvtemperature,n,m}$ is isometric from $\dom C_n$ with the covariance norm into $\sphilbfrak{H}_{\sminvtemperature,C_m}$, and
$V_m\iota_{\sminvtemperature,n,m}=C_m^{\frac{1}{2}}\iota_{\sminvtemperature,n,m}=\iota_{\sminvtemperature,n,m}V_n$ on $\dom C_n$ extends by
continuity.
By \eqref{eq:sharp-time-vector} and the intertwining of the functional calculus,
$\rbk{V_mj_{m,t}\tilde{\iota}_{n,m}g}_k=\tilde{\iota}_{n,m}\rbk{V_nj_{n,t}g}_k$, so
$V_mj_{m,t}\tilde{\iota}_{n,m}g=\iota_{\sminvtemperature,n,m}V_nj_{n,t}g=V_m\overline{\iota}_{\sminvtemperature,n,m}j_{n,t}g$, and $V_m$ is
injective.
The remaining relations hold on $\dom C_n$, because $\iota_{\sminvtemperature,n,m}$ acts pointwise in time, and extend by continuity.
\end{proof}

\begin{defn}[Compatible regularizers]\label{def:compatible-regularizers}
Invariant regularizers $\opform{R}_{\sminvtemperature,C_n}$, $n\in\monnat$, in the sense of Definition \ref{def:invariant-regularizer} are
compatible if $\overline{\iota}_{\sminvtemperature,n,m}\dom\opform{R}_{\sminvtemperature,C_n}\subset\dom\opform{R}_{\sminvtemperature,C_m}$ and
$\opform{R}_{\sminvtemperature,C_m}\overline{\iota}_{\sminvtemperature,n,m}=\overline{\iota}_{\sminvtemperature,n,m}\opform{R}_{\sminvtemperature,C_n}$ on
$\dom\opform{R}_{\sminvtemperature,C_n}$ for $n\leq m$.
For them, the objects of Section \ref{sec:markov} for $\opform{R}_{\sminvtemperature,C_n}$ carry the subscript $n$:
$\prbqspace_{\pm,n}$, $\tilde{u}_{n,\tau}$, $\tilde{r}_n$, $\msrprb_{\sminvtemperature,0,n}$, $\opfocksegal_n$, $\opfocksegal_{n,t}$,
$\mblfmlfrak{S}_{n,I}$, $U_{n,\tau}$, and $R_n$.
\end{defn}

\begin{thm}[Quasi-local construction]\label{thm:quasi-local}
Let the regularizers be compatible.
In (1) and (2), let $n\leq m\leq l$.
\begin{enumerate}
\item The adjoint $\faadj{\overline{\iota}_{\sminvtemperature,n,m}}$ extends uniquely to a continuous linear map
$\pi_{\sminvtemperature,m,n}\colon\prbqspace_{-,m}\to\prbqspace_{-,n}$, and, for $q\in\prbqspace_{-,m}$, $f\in\prbqspace_{+,n}$, and
$\tau\in\fldreal$,
\begin{equation}\label{eq:quasi-local-maps}
\begin{aligned}
\bkt{\pi_{\sminvtemperature,m,n}q}{f}_{\prbqspace_{-,n},\prbqspace_{+,n}}
&=
\bkt{q}{\overline{\iota}_{\sminvtemperature,n,m}f}_{\prbqspace_{-,m},\prbqspace_{+,m}},
\\ 
\pi_{\sminvtemperature,n,n}
&=
1,
\quad
\pi_{\sminvtemperature,l,n}
=
\pi_{\sminvtemperature,m,n}\circ\pi_{\sminvtemperature,l,m},
\\ 
\pi_{\sminvtemperature,m,n}\circ\tilde{u}_{m,\tau}
&=
\tilde{u}_{n,\tau}\circ\pi_{\sminvtemperature,m,n},
\quad
\pi_{\sminvtemperature,m,n}\circ\tilde{r}_m
=
\tilde{r}_n\circ\pi_{\sminvtemperature,m,n}.
\end{aligned}
\end{equation}
\item The measures, the fields, the $\sigma$-algebras, and the automorphisms are consistent: for
$v\in\sphilbfrak{H}_{\sminvtemperature,C_n}$, $t\in S_{\sminvtemperature}$, $g\in\mathcal{T}_{\fldreal,n}$, $I\subset S_{\sminvtemperature}$,
$F\in\fun{\lp^2}{\prbqspace_{-,n},\msrprb_{\sminvtemperature,0,n}}$, and $\tau\in\fldreal$,
\begin{equation}\label{eq:quasi-local-consistency}
\begin{aligned}
\pushoutrbk{\pi_{\sminvtemperature,m,n}}\msrprb_{\sminvtemperature,0,m}
&=
\msrprb_{\sminvtemperature,0,n},
\\ 
\fun{\opfocksegal_n}{v}\circ\pi_{\sminvtemperature,m,n}
&=
\fun{\opfocksegal_m}{\overline{\iota}_{\sminvtemperature,n,m}v},
\\ 
\fun{\opfocksegal_{n,t}}{g}\circ\pi_{\sminvtemperature,m,n}
&=
\fun{\opfocksegal_{m,t}}{\tilde{\iota}_{n,m}g},
\\ 
\set{\fun{\inv{\pi_{\sminvtemperature,m,n}}}{A}}{A\in\mblfmlfrak{S}_{n,I}}
&\subset
\mblfmlfrak{S}_{m,I},
\\ 
U_{m,\tau}\rbk{F\circ\pi_{\sminvtemperature,m,n}}
&=
\rbk{U_{n,\tau}F}\circ\pi_{\sminvtemperature,m,n},
\\ 
R_m\rbk{F\circ\pi_{\sminvtemperature,m,n}}
&=
\rbk{R_nF}\circ\pi_{\sminvtemperature,m,n}.
\end{aligned}
\end{equation}
\item Let $\prbqspace_{-,\infty}$ be the set of $\seq{q_n}{n\in\monnat}\in\prod_n\prbqspace_{-,n}$ with
$\pi_{\sminvtemperature,m,n}q_m=q_n$ for $n\leq m$, with the coordinate maps $p_n$ and the $\sigma$-algebra generated by them.
There is a unique probability measure $\msrprb_{\sminvtemperature,0,\infty}$ on $\prbqspace_{-,\infty}$ with
$\pushoutrbk{p_n}\msrprb_{\sminvtemperature,0,\infty}=\msrprb_{\sminvtemperature,0,n}$ for every $n$, and the coordinatewise maps
$\tilde{u}_{\infty,\tau}=\prod_n\tilde{u}_{n,\tau}$ and $\tilde{r}_\infty=\prod_n\tilde{r}_n$ map $\prbqspace_{-,\infty}$ into itself
and preserve $\msrprb_{\sminvtemperature,0,\infty}$.
\end{enumerate}
The algebras $\set{F\circ p_n}{F\in\fun{\lp^\infty}{\prbqspace_{-,n},\msrprb_{\sminvtemperature,0,n}}}$ increase with $n$; we call their
union the algebra of local observables and its norm closure in
$\fun{\lp^\infty}{\prbqspace_{-,\infty},\msrprb_{\sminvtemperature,0,\infty}}$ the quasi-local algebra, which is invariant under
$\tilde{u}_{\infty,\tau}$ and $\tilde{r}_\infty$.
\end{thm}

\begin{proof}
(1): Compatibility gives the first of the following identities, and taking adjoints gives the second:
$$\begin{aligned}
\opform{R}_{\sminvtemperature,C_m}^{-1}\overline{\iota}_{\sminvtemperature,n,m}
&=
\overline{\iota}_{\sminvtemperature,n,m}\opform{R}_{\sminvtemperature,C_n}^{-1},
\\ 
\opform{R}_{\sminvtemperature,C_n}^{-1}\faadj{\overline{\iota}_{\sminvtemperature,n,m}}
&=
\faadj{\overline{\iota}_{\sminvtemperature,n,m}}\opform{R}_{\sminvtemperature,C_m}^{-1}.
\end{aligned}$$
Hence $\norm{\faadj{\overline{\iota}_{\sminvtemperature,n,m}}v}_{\prbqspace_{-,n}}\leq\norm{v}_{\prbqspace_{-,m}}$ for
$v\in\sphilbfrak{H}_{\sminvtemperature,C_m}$, and $\faadj{\overline{\iota}_{\sminvtemperature,n,m}}$ extends by density.
Since $\overline{\iota}_{\sminvtemperature,n,m}f\in\prbqspace_{+,m}$, the pairing identity holds on $\sphilbfrak{H}_{\sminvtemperature,C_m}$ by
the definition of the adjoint and extends by continuity.
The other relations follow from Lemma \ref{lem:covariance-embeddings} by taking adjoints, with $\faadj{u_\tau}=u_{-\tau}$ and
$\faadj{r}=r$, and by continuity, since $\tilde{u}_{n,\tau}$ and $\tilde{r}_n$ are the continuous extensions of $u_\tau$ and $r$.

(2): By \eqref{eq:quasi-local-maps},
$\fun{\opfocksegal_n}{f}\circ\pi_{\sminvtemperature,m,n}=\fun{\opfocksegal_m}{\overline{\iota}_{\sminvtemperature,n,m}f}$ everywhere for
$f\in\prbqspace_{+,n}$.
With the isometry of $\overline{\iota}_{\sminvtemperature,n,m}$, the image measure
$\pushoutrbk{\pi_{\sminvtemperature,m,n}}\msrprb_{\sminvtemperature,0,m}$ therefore has the characteristic functional
\eqref{eq:characteristic-functional} of $\msrprb_{\sminvtemperature,0,n}$, and the uniqueness in Proposition \ref{prop:gaussian-measure} gives
the first line.
Composition with $\pi_{\sminvtemperature,m,n}$ is then isometric on $\lp^2$, so the field identity extends from $\prbqspace_{+,n}$ to
$\sphilbfrak{H}_{\sminvtemperature,C_n}$, and the second line of \eqref{eq:covariance-embeddings} gives the sharp-time identity, with
$\tilde{\iota}_{n,m}g\in\mathcal{T}_{\fldreal,m}$ by \eqref{eq:directed-system}.
The inclusion of the $\sigma$-algebras follows from the sharp-time identity for the generators and from the first line for the null sets, and the
last two lines follow from \eqref{eq:path-space-automorphisms} and \eqref{eq:quasi-local-maps}.

(3): The space $\prbqspace_{-,n}$ is a separable Hilbert space, hence Polish.
For $n_1<\cdots<n_N$, the image of $\msrprb_{\sminvtemperature,0,n_N}$ under $q\mapsto\seq{\pi_{\sminvtemperature,n_N,n_i}q}{i}$ defines
consistent finite-dimensional distributions by \eqref{eq:quasi-local-maps} and \eqref{eq:quasi-local-consistency}, so the Kolmogorov
theorem \cite[Vol.~II, Corollary 7.7.2, p.~96]{VladimirBogachev001} gives a probability measure on $\prod_n\prbqspace_{-,n}$ with these
distributions.
Each closed set $\set{\seq{q_l}{l}}{\pi_{\sminvtemperature,m,n}q_m=q_n}$ has measure one, so the measure is carried by
$\prbqspace_{-,\infty}$; it is unique because, on $\prbqspace_{-,\infty}$, every cylinder set is of the form $p_{n_N}^{-1}\rbk{B}$ with $B$
Borel, so the cylinder sets form a $\pi$-system generating the $\sigma$-algebra.
By \eqref{eq:quasi-local-maps}, $\tilde{u}_{\infty,\tau}$ and $\tilde{r}_\infty$ preserve $\prbqspace_{-,\infty}$ and satisfy
$p_n\circ\tilde{u}_{\infty,\tau}=\tilde{u}_{n,\tau}\circ p_n$ and $p_n\circ\tilde{r}_\infty=\tilde{r}_n\circ p_n$; by
Proposition \ref{prop:gaussian-measure} the image measures have the same marginals, so they equal
$\msrprb_{\sminvtemperature,0,\infty}$, and the local algebras, hence the quasi-local algebra, are invariant.
The local algebras increase because $F\circ p_n=\rbk{F\circ\pi_{\sminvtemperature,m,n}}\circ p_m$ for $n\leq m$.
\end{proof}

\begin{ex}[Nested momentum lattices]\label{ex:momentum-lattices}
Let $d\in\semigrposint$, $s>0$, $\smchemicalpotential<0$, and $a>0$, let $\seq{M_n}{n\in\monnat}$ be a strictly increasing sequence of
positive integers with $M_n\mid M_{n+1}$, and set $L_n=M_na$.
Let $\sphilbfrak{h}_n=\fun{\lpseq^2}{L_n^{-1}\ringratint^d;\fldcmp}$, let $\physham[h]_n$ be the multiplication by
$\abs{2\pi k}^s-\smchemicalpotential$, let $\tilde{\iota}_{n,m}$ be the extension by zero, and let $\mathcal{T}_{\fldreal,n}$ consist of the
finitely supported sequences.
For the free Bose gas on the cube of side $L_n$ with periodic boundary conditions, $\physham[h]_n$ is the one-particle Hamiltonian at the
chemical potential $\smchemicalpotential$ in the Fourier coefficients.
The spectral projections of $\physham[h]_n$ are multiplications by $\fun{\fndef{E}}{\abs{2\pi k}^s-\smchemicalpotential}$, which commute with
the extension by zero, and the operators $\fun{\eta_{i,\sigma}}{\physham[h]_n}$ preserve the finitely supported sequences, so
Definition \ref{def:directed-system} applies.
Let $\opform{R}_{\sminvtemperature,C_n}$ be diagonal in the vectors $V_n^{-1}\rbk{c_j\otimes\psi}$, $j\in\monnat$, and
$V_n^{-1}\rbk{s_j\otimes\psi}$, $j\in\semigrposint$, with $\psi\in\setone{\kroneckerdelta_k,\imunit\kroneckerdelta_k}$, with the eigenvalues
$\rbk{1+j}\rbk{1+\abs{2\pi k}^2}^d$.
These blocks are those of Lemma \ref{lem:regularizer-existence}, and
$\sum_k\rbk{1+\abs{2\pi k}^2}^{-2d}<\infty$, so $\opform{R}_{\sminvtemperature,C_n}$ is an invariant regularizer.
The first line of \eqref{eq:covariance-embeddings} gives
$\overline{\iota}_{\sminvtemperature,n,m}V_n^{-1}\rbk{c_j\otimes\psi}=V_m^{-1}\rbk{c_j\otimes\psi}$, and likewise for $s_j\otimes\psi$, so
these regularizers are compatible, and Theorem \ref{thm:quasi-local} applies.
Since $\physham[h]_n\geq-\smchemicalpotential>0$, each $\physham[h]_n$ has a spectral gap, so this example lies within the setting of
\cite[Section 21.4]{DerezinskiGerard001}; Theorem \ref{thm:quasi-local} itself does not use a gap.
The embeddings extend momentum sequences by zero; they are not restrictions to nested spatial regions, and the limit
$\msrprb_{\sminvtemperature,0,\infty}$ is not an infinite-volume field.
\end{ex}

\bibliography{myref.bib}

@article{
    GerardJakel001,
    author = {Christian Gérard and Christian D. Jäkel},
    title = {{Thermal quantum fields with spatially cutoff interactions in 1+1 space--time dimensions}},
    journal = "Journal of Functional Analysis",
    volume = "220",
    number = "1",
    pages = "157--213",
    year = "2005",
    doi = "10.1016/j.jfa.2004.08.003",
    arxiv = "math-ph/0307053",
    url = "https://doi.org/10.1016/j.jfa.2004.08.003",
}

@article{
    ArakiWoods001,
    author = {H. Araki and E. J. Woods},
    journal = "J. Math. Phys.",
    pages = "637-662",
    tag = {mathematical physics, mathematics, physics, quantum statistical physics},
    title = {{Representations of the canonical commutation relations describing a nonrelativistic infinite free Bose gas}},
    volume = "4",
    year = "1963",
}

@article{
    BuchholzGrundling002,
    author = {Detlev Buchholz and Hendrik Grundling},
    journal = "J. Funct. Anal.",
    month = "June",
    pages = "2725-2779",
    tag = {mathmatical physics, operator algebra},
    title = {{The Resolvent Algebra: A New Approach to Canonical Quantum Systems}},
    volume = "254",
    year = "2008",
}

@book{
    DerezinskiGerard001,
    author = {Jan Dereziński and Christian Gérard},
    pages = "688",
    publisher = "Cambridge University Press",
    tag = {mathematical physics},
    title = {{Mathematics of Quantization and Quantum Fields}},
    year = "2022",
}

@article{
    DerezinskiJaksic002,
    author = {J. Derezi\'{n}ski and V. Jak\v{s}i\'{c}},
    journal = "Ann. Henri Poincar\'{e}",
    pages = "739-793",
    tag = {mathematical physics, mathematics, operator algebra, spectral theory},
    title = {{Return to Equilibrium for Pauli-Fierz Systems}},
    volume = "4",
    year = "2003",
}

@book{
    GeraldFolland001,
    author = {Gerald B. Folland},
    month = "4",
    pages = "416",
    publisher = "Wiley",
    tag = {mathematics, analysis},
    title = {{Real Analysis Modern Techniques and Their Applications}},
    year = "1999",
}

@article{
    KleinLandau001,
    author = {Abel Klein and Lawrence J. Landau},
    journal = "J. Funct. Anal.",
    month = "6",
    pages = "368-428",
    tag = {mathematics, mathematical physics, operator algebra, probability},
    title = {{Stochastic Processes Associated with KMS States}},
    volume = "42",
    year = "1981",
}

@book{
    LorincziHiroshimaBetz003,
    author = {J. Lörinczi and F. Hiroshima},
    month = "3",
    pages = "539",
    publisher = "Walter De Gruyter",
    tag = {mathematical physics, probability, operator theory, quantum mechanics, statistical mechanics, quantum field theory, path integral, functional integral},
    title = {{Feynman-Kac-Type Theorems and Gibbs Measures on Path Space: Applications in Rigorous Quantum Field Theory (2)}},
    volume = "2",
    year = "2020",
}

@article{
    VIYukalov001,
    author = {V. I. Yukalov},
    journal = "Laser Phys.",
    month = "June",
    pages = "1145-1168",
    tag = {laser physics},
    title = {{Difference in Bose-Einstein condensation of conserved and unconserved particles}},
    url = "https://link.springer.com/article/10.1134/S1054660X12070171",
    volume = "22",
    year = "2012",
}

@book{
    VladimirBogachev001,
    author = {Vladimir I. Bogachev},
    address = "Berlin",
    doi = "10.1007/978-3-540-34514-5",
    month = "11",
    pages = "1075",
    publisher = "Springer",
    tag = {mathematics, measure theory, integration, probability},
    title = {{Measure Theory}},
    url = "https://doi.org/10.1007/978-3-540-34514-5",
    year = "2006",
}

@book{
    VladimirBogachev002,
    author = {V. I. Bogachev},
    address = "Providence, RI",
    doi = "10.1090/surv/062",
    pages = "xii+433",
    publisher = "American Mathematical Society",
    series = "Mathematical Surveys and Monographs",
    tag = {probability, Gaussian measures},
    title = {{Gaussian measures}},
    url = "https://doi.org/10.1090/surv/062",
    volume = "62",
    year = "1998",
}

\end{document}